\documentclass[letterpaper]{article} % DO NOT CHANGE THIS
\usepackage{aaai25}  % DO NOT CHANGE THIS
\usepackage{times}  % DO NOT CHANGE THIS
\usepackage{helvet}  % DO NOT CHANGE THIS
\usepackage{courier}  % DO NOT CHANGE THIS
\usepackage[hyphens]{url}  % DO NOT CHANGE THIS
\usepackage{graphicx} % DO NOT CHANGE THIS
\usepackage{natbib}  % DO NOT CHANGE THIS AND DO NOT ADD ANY OPTIONS TO IT
\usepackage{caption} % DO NOT CHANGE THIS AND DO NOT ADD ANY OPTIONS TO IT
\usepackage{algorithm}
\usepackage{algpseudocode}
\usepackage{amsmath,amssymb}

\usepackage{newfloat}
\usepackage{listings}
\DeclareCaptionStyle{ruled}{labelfont=normalfont,labelsep=colon,strut=off} % DO NOT CHANGE THIS
\floatstyle{ruled}
\newfloat{listing}{tb}{lst}{}
\floatname{listing}{Listing}
\usepackage{amsmath}
\usepackage{booktabs}
\usepackage{nicefrac}

\usepackage{amsthm} % For theorem environments
\newtheorem{Assump}{Assumption}
\newtheorem{Prop}{Proposition}
\newtheorem{lemma}{Lemma}

\title{The Impact of Cut Layer Selection in Split Federated Learning}
\author{
    Justin Dachille\textsuperscript{\rm 1},
    Chao Huang\textsuperscript{\rm 2},
    Xin Liu\textsuperscript{\rm 1}
}
\affiliations{
    \textsuperscript{\rm 1}University of California, Davis\\
    \textsuperscript{\rm 2}Montclair State University\\
    jtdachille@ucdavis.edu, huangch@montclair.edu, xinliu@ucdavis.edu
}

\begin{document}
\maketitle
\begin{abstract}
Split Federated Learning (SFL) is a distributed machine learning paradigm that combines federated learning and split learning. In SFL, a neural network is partitioned at a cut layer, with the initial layers deployed on clients and remaining layers on a training server. There are two main variants of SFL: SFL-V1 where the training server maintains separate server-side models for each client, and SFL-V2 where the training server maintains a single shared model for all clients. While existing studies have focused on algorithm development for SFL, a comprehensive quantitative analysis of how the cut layer selection affects model performance remains unexplored. This paper addresses this gap by providing numerical and theoretical analysis of SFL performance and convergence relative to cut layer selection. We find that SFL-V1 is relatively invariant to the choice of cut layer, which is consistent with our theoretical results. Numerical experiments on four datasets and two neural networks show that the cut layer selection significantly affects the performance of SFL-V2. Moreover, SFL-V2 with an appropriate cut layer selection outperforms FedAvg on heterogeneous data.
\end{abstract}

\section{Introduction}
Federated Learning (FL) is a popular approach to collectively train machine learning models while preserving data privacy among clients such as mobile phones, IoT devices, and edge devices. \cite{mcmahan2017communication, brisimi2018federated}. In conventional FL, clients train models in parallel and then upload their models to a coordinating server for aggregation. However, FL faces significant computational and communication challenges in resource-constrained environments, as each client needs to train the entire model on-premise and communicate the full model to the server \cite{pmlr-v119-hamer20a, caldas2018expanding}. This limitation becomes particularly critical as modern machine learning models, especially Large Language Models (LLMs), continue to grow to billions of parameters and beyond \cite{minaee2024large}.

To address these limitations, Split Learning (SL) is a promising solution \cite{gupta2018distributed, vepakomma2018split}. In SL, a neural network is typically split into two parts where the clients train the initial layers of the neural network (e.g., feature extraction layers) and send activations to the server. The server completes the training with the remaining layers and sends back gradients which the client uses to finish back propagation on its client side network. By splitting the model, SL reduces the computational burden on clients. Additionally, studies have shown that SL can outperform FL in terms of communication efficiency as the number of clients increases, since only layer activations rather than full model parameters need to be transmitted \cite{singh2019detailed}. However, the sequential relay-based training of clients results in prohibitively long training times and inherently reduced scalability. Moreover, SL suffers from catastrophic forgetting, where the model tends to forget previously learned features when training on new client data \cite{duan2022combined}.

These limitations motivated the development of Split Federated Learning (SFL) \cite{thapa2022splitfed}. SFL aims to combine the advantages of both FL and SL while minimizing their respective drawbacks. By leveraging parallel training from FL and reduced client-side computation from SL, SFL offers a promising solution for efficient and privacy-preserving distributed learning in resource-constrained environments. We summarize the differences between these three algorithms in Fig.~\ref{fig:fl-differences}.

\begin{figure*}[t]
    \centering
    \includegraphics[width=0.9\textwidth]{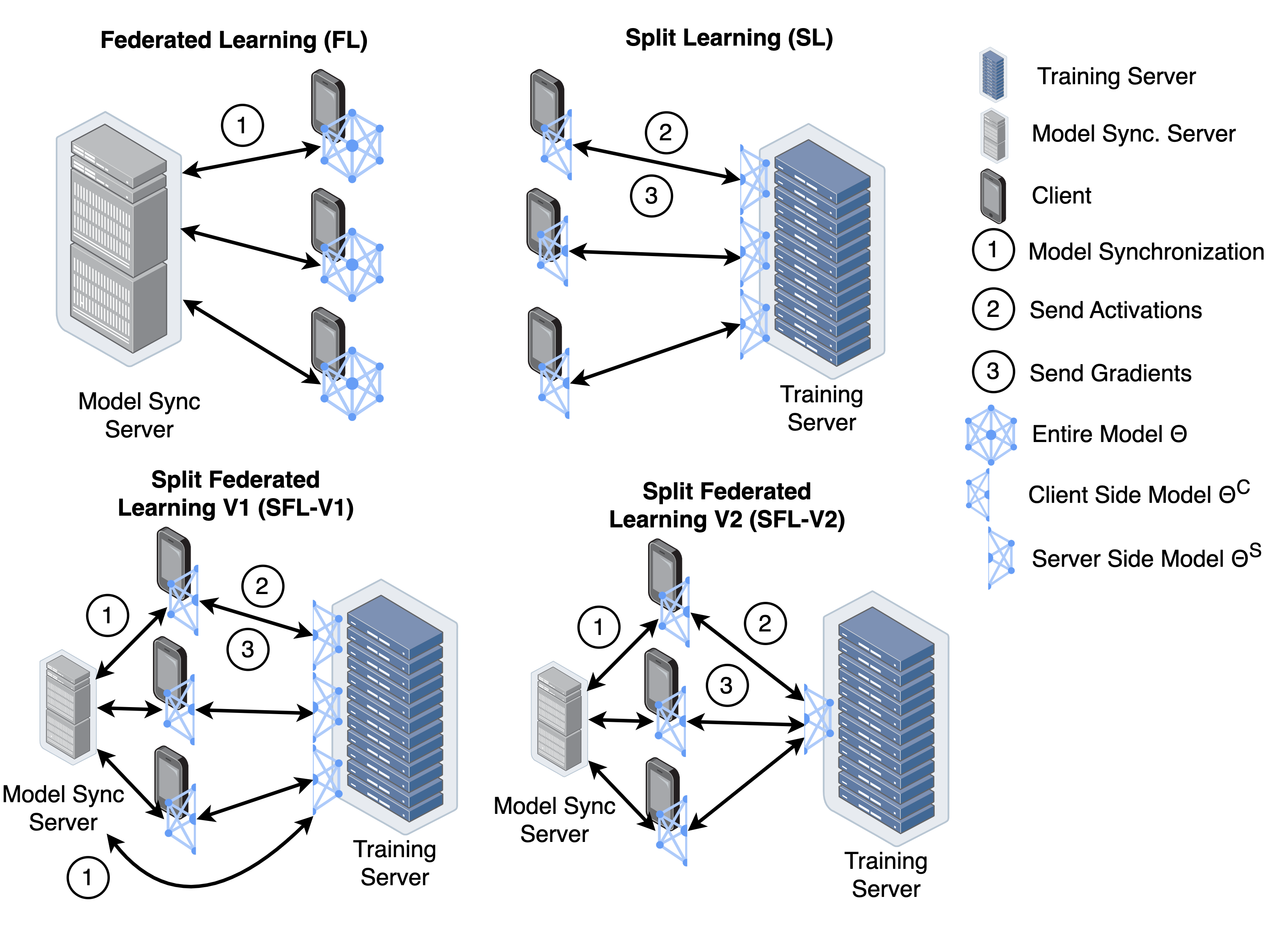}
    \caption{Comparison of distributed learning architectures. The Model Sync Server maintains model consistency across clients, while the Training Server handles model computations. In FL (top left), clients train complete model copies locally and periodically synchronize with the Model Sync Server (1). In SL (top right), the model $\theta$ is split at a cut layer into client-side ($\theta^c$) and server-side ($\theta^s$) components, where clients sequentially take turns: each client computes forward activations up to the cut layer (2), the server completes the forward pass and backpropagation to the cut layer, and the client finishes the backward pass using returned gradients (3). SFL comes in two variants: SFL-V1 (bottom left) and SFL-V2 (bottom right). Both variants split the model and maintain client-side synchronization through FedAvg (1), but differ in server-side processing: after clients send activations (2) and receive gradients (3), SFL-V1 aggregates both client and server-side models, whereas SFL-V2 only aggregates client-side models.}
    \label{fig:fl-differences}
\end{figure*}

Despite recent advances in SFL algorithms, a critical yet understudied aspect is the selection of the cut layer—the point at which the neural network is split between client and server. Prior work on cut layer selection in SFL has primarily focused on communication efficiency and privacy preservation \cite{wu2023split, zhang2023privacy}. However, the impact of cut layer selection on model convergence and  performance remains under-explored. This is a critical gap because the cut layer determines the computational load distribution and can affect model performance. Understanding these dynamics is essential for deploying SFL effectively in real-world applications.

In this paper, we aim to address two fundamental questions regarding SFL:

\begin{itemize}
    \item \textbf{Question 1:} How does the cut layer affect the performance of SFL?
    \item \textbf{Question 2:} How does the performance of SFL compare to FL?
\end{itemize}

These questions are motivated by several critical factors. First, deep learning is resource-intensive to implement on edge devices \cite{9220170}, and the choice of cut layer can significantly affect computation requirements. Recent work by Lee et al. \cite{lee2024exploring} demonstrates that deeper cut layers increase both client-side computation and communication energy costs - client computation increases with more local layers, while communication costs rise from transmitting larger model parameters during synchronization. They also show that deeper cut layers enhance privacy by making input reconstruction from smashed data more difficult, highlighting the trade-off between efficiency and privacy in cut layer selection. Second, while SFL shows promise as a distributed learning approach, its performance characteristics, particularly in relation to cut layer selection and data heterogeneity, remain underexplored.

Regarding the impact of cut layer selection, we prove theoretically and validate experimentally that SFL-V1's performance remains invariant to cut layer placement in both Independent and Identically Distributed (IID) and non-IID settings. This behavior stems from SFL-V1's architecture, where separate server-side models for each client preserve the independence of client training paths regardless of cut layer location. In contrast, SFL-V2 exhibits significant performance variations respect to with cut layer placement, which we postulate because its shared server-side model enables effective learning from all client activations simultaneously.

% In comparing SFL variants with FL, we find that SFL-V2 with an early cut layer significantly outperforms FedAvg on heterogeneous data. SFL-V1, while maintaining consistent performance across cut layers, achieves lower accuracy compared to SFL-V2 and similar to FL.

Our key contributions are summarized below:

\begin{itemize}
    \item We provide a quantitative study on the effect of cut layer selection on SFL performance.
    
    \item We provide a convergence analysis of SFL with respect to cut layer selection. We show that SFL-V1 is invariant to cut layer selection.
    
    \item We conduct numerical experiments across four datasets and two model architectures. We find that SFL-V2 with an early cut layer significantly outperforms FedAvg on heterogeneous data. SFL-V1, while maintaining consistent performance across cut layers, achieves lower accuracy compared to SFL-V2 and similar to FL. 
    % demonstrating that SFL-V2 outperforms FL in both IID and non-IID scenarios.
\end{itemize}

\section{Related Work}

\subsection{Federated Learning}

Federated Learning (FL) was introduced by McMahan et al. \cite{mcmahan2017communication} as a distributed learning paradigm where multiple clients collaborate to train a model while keeping their data local. In FL, each client trains a complete copy of the model on their local data, and a model synchronization server aggregates these models to create a global model. The most widely-used FL algorithm, FedAvg, performs weighted averaging of client models based on their local dataset sizes.

Data heterogeneity represents one of the fundamental challenges in FL.
% federated learning (FL), first introduced by McMahan et al. \cite{mcmahan2017communication}. In FL, multiple clients collaborate to train a model while keeping their data local, with the central server aggregating client models through weighted averaging (FedAvg). However, 
When clients have heterogeneous data distributions, FedAvg can suffer significant performance degradation due to the notorious client drift issue
% - up to 55\% accuracy drops compared to centralized training 
\cite{hsu2019measuring}. Several approaches have been proposed to address this challenge: SCAFFOLD \cite{karimireddy2020scaffold} uses control variates to correct for client drift in non-IID settings, while FedPer \cite{arivazhagan2019federated} maintains personalized layers on clients while sharing base layers globally. Other approaches like FedDF \cite{lin2020ensemble} employ ensemble distillation for better heterogeneous model aggregation, and FedProx \cite{li2020federated} adds a proximal term to limit local updates from diverging from the global model. Recent works have introduced novel approaches to tackle heterogeneity: FedUV \cite{son2024feduv} introduces regularization terms to emulate IID settings locally by controlling classifier variance and representation uniformity. MOON \cite{li2021model} leverages model-level contrastive learning to correct local training. FedAlign \cite{10371403} addresses data heterogeneity through data-free knowledge distillation, using a generator to estimate global feature distributions and align local models accordingly.\footnote{While these techniques were primarily designed for FL, they can be easily implemented in SFL.}

% Our results show that SFL-V2 with an early cut layer can significantly outperform FedAvg on non-IID data - achieving 52.38\% accuracy versus 42.60\% on CIFAR-100. This improvement suggests that SFL's hybrid architecture may offer inherent advantages for handling heterogeneous data distributions compared to traditional FL approaches.

\subsection{Split Learning}
Split Learning (SL), introduced by Gupta and Raskar \cite{gupta2018distributed}, takes a different approach to distributed learning. Instead of training complete models on clients, SL splits a neural network at a cut layer, with initial layers on clients and remaining layers on the server. In SL, clients process data through their local layers up to the cut layer, then send activations to the server. The server completes the forward pass, computes gradients, and sends these gradients back to clients for local model updates.

Several works have explored SL's potential for handling heterogeneous data distributions. SplitAVG \cite{zhang2022splitavg} combines network splitting with feature map concatenation to train an unbiased estimator of the target data distribution. Li and Lyu \cite{li2024convergence} provided theoretical convergence guarantees for split learning, suggesting potential advantages over FedAvg on heterogeneous data.
% Han et al. \cite{han2024convergence} provided theoretical convergence guarantees for Sequential SL on heterogeneous data, suggesting potential advantages over FedAvg in such settings. 
COMSPLIT \cite{ninkovic2024comsplit} introduced a communication-aware SL framework for time series data, incorporating early-exit strategies to handle devices with heterogeneous computational capabilities. FedCST \cite{wang2024fedcst} proposed a hybrid approach combining device clustering with SL to address both data and device heterogeneity. AdaSplit \cite{chopra2021adasplit} focused on improving SL's performance across heterogeneous clients while reducing bandwidth consumption through adaptive mechanisms. In the domain of graph neural networks, SplitGNN \cite{xu2023splitgnn} addressed heterogeneity in graph data through a split architecture with heterogeneous attention mechanisms. RoofSplit \cite{huang2023roofsplit} proposed an edge computing framework that optimally splits CNN models across heterogeneous edge nodes based on Roofline theory \cite{williams2009roofline}.

\subsection{Split Federated Learning}
Split Federated Learning (SFL) was originally proposed by Thapa et al. \cite{thapa2022splitfed} as a hybrid approach combining FL and SL. Since its introduction, several works have explored various aspects of SFL, from privacy and security to system optimization and scalability.

Privacy preservation has been a key focus in SFL research. Li et al. \cite{li2022ressfl} proposed ResSFL, a framework designed to be resistant to model inversion attacks during training through attacker-aware training of the feature extractor. Khan et al. \cite{khan2022security} conducted the first empirical analysis of SFL's robustness against model poisoning attacks, demonstrating that SFL's lower-dimensional model updates provide inherent robustness compared to traditional FL. Zhang et al. \cite{zhang2023privacy} demonstrated that deeper cut layers in SFL improve privacy by reducing reconstruction accuracy at the cost of increased client computation and communication overhead. Zheng et al. \cite{zheng2024ppsfl} introduced PPSFL, incorporating private group normalization layers to protect privacy while addressing data heterogeneity.

System optimization and efficiency have been another major research direction. Mu et al. \cite{mu2023communication} developed CSE-FSL to reduce communication overhead and server storage requirements through an auxiliary network and selective epoch communication. Gao et al. \cite{gao2024pipesfl} proposed PipeSFL, a fine-grained parallelization framework that addresses client heterogeneity through priority scheduling and hybrid training modes. Xu et al. \cite{xu2023accelerating} tackled the challenge of heterogeneous devices by jointly optimizing cut layer selection and bandwidth allocation.

Several works have focused on enhancing SFL's performance and applicability. Yang et al. \cite{yang2022robust} introduced RoS-FL specifically for U-shaped medical image networks, incorporating a dynamic weight correction strategy to stabilize training and prevent model drift. Li et al. \cite{li2023split} demonstrated SFL's effectiveness in healthcare settings, achieving comparable performance to centralized and federated approaches while improving privacy and computational efficiency. Shin et al. \cite{shin2023fedsplitx} proposed FedSplitX to handle system heterogeneity in foundation models by implementing multiple partition points and auxiliary networks.
% Recent innovations have focused on addressing both system and statistical heterogeneity. 
Liao et al. \cite{liao2024mergesfl} developed MergeSFL, which combines feature merging and batch size regulation to improve accuracy and training efficiency. Zhu et al. \cite{zhu2024esfl} proposed ESFL, which jointly optimizes user-side workload and server-side resource allocation for heterogeneous environments.

Recent work has begun to examine the implications of cut layer selection, though primarily from system-level perspectives. Lee et al. \cite{lee2024exploring} provided a comprehensive analysis of how cut layer placement affects both energy consumption and privacy preservation in SFL systems, demonstrating that cut layer selection creates important trade-offs between computational efficiency and data protection. Shiranthika et al. \cite{shiranthika2023splitfed} investigated the robustness of different cut layer positions against network packet loss, finding statistical evidence that deeper split points can provide advantages in certain applications. While these works provide valuable insights into system-level tradeoffs, \textbf{they do not address the question of how cut layer selection affects model performance}. Note that recent work has examined SFL convergence properties \cite{han2024convergence}, but it does not analyze the specific effects of cut layer selection. Our work addresses this gap by providing the first comprehensive analysis of how cut layer selection affects model performance across various datasets, architectures, and data distributions.

% % \begin{figure*}[t]
% %     \centering
% %     \includegraphics[width=0.9\textwidth]{SFLVersions.png}
% %     \caption{Comparison of Split Federated Learning Version 1 (SFL-V1) and Version 2 (SFL-V2) architectures. Both variants split the model $\theta$ into client-side ($\theta^c$) and server-side ($\theta^s$) components. In SFL-V1 (left), after clients send their activations to the Training Server, server-side models are executed in parallel and then aggregated. In SFL-V2 (right), the server processes client activations sequentially through a single server-side model without aggregation. In both variants, clients synchronize their client-side models through periodic FedAvg, but differ in server-side processing: SFL-V1 maintains separate server models that are aggregated, while SFL-V2 uses a single server model updated sequentially with each client's data.}
% %     \label{fig:sfl_variants}
% \end{figure*}

\section{Problem Formulation}
We provide an overview of SFL here. The complete algorithmic description with detailed notation can be found in Appendix~\ref{appendix:algorithms}. Consider a set of clients $\mathcal{K} = \{1, 2, \ldots, K\}$, where each client $k \in \mathcal{K}$ possesses a local private dataset $\mathcal{D}_k$ of size $D_k=|\mathcal{D}_k|$. The global model, parameterized by $\theta$, consists of $L$ layers. In SFL, the global model is split at the $L_c$-th layer (the cut layer) into two segments: 
\begin{itemize}
    \item Client-side model $\theta^C$ (layers 1 to $L_c$),
    \item Server-side model $\theta^S$ (layers $L_c + 1$ to $L$),
\end{itemize}
where $\theta = \{\theta^C; \theta^S\}$. Let $\theta^C_k$ denote client $k$'s local client-side model. The system involves two servers (see Fig. \ref{fig:fl-differences}): 
\begin{itemize}
    \item Model synchronization server (MSS): It aggregates and distributes model weights, similar to FL.
    \item Training server (TS): It maintains server-side models and performs training computations.
\end{itemize}
In this paper, we focus on two variants of SFL as proposed in \cite{thapa2022splitfed}: SFL-V1 and SFL-V2. The key distinction between these variants lies in how the TS manages its models. In SFL-V1, the TS maintains separate server-side models $\theta^S_k$ for each client $k$, which the MSS aggregates after each communication round to form a global server-side model. In SFL-V2, the TS maintains a single shared model $\theta^S$ that processes data from all clients sequentially. 

More specifically,  let $F_k (\theta, B_k)$ denote the loss of model $\theta$ over client $k$'s mini-batch samples $B_k$, which is randomly sampled from client $k$'s dataset $\mathcal{D}_k$. %, computed on $\theta$, over the instance $\zeta_n$. 
Let $ F_k (\theta)\triangleq \mathbb{E}_{B_k\sim \mathcal{D}_k}[ F_k (\theta, B_k)]$ denote the expected loss of model $\theta$ over client $k$'s dataset. SFL aims to minimize the expected loss of the model $\theta$ over the datasets of all clients:
\begin{equation}\label{global-loss}
\min_{\theta} f(\theta)=\sum_{k=1}^K \alpha_k F_k\left(\theta\right), 
\end{equation}
where $\alpha_k\in [0,1]$ is the weight of client $k$, and we typically have $\alpha_k=D_k/\sum_{k'\in \mathcal{K}} D_{k'}$.

An SFL process operates for $T$ communication rounds to minimize the global loss function $f(\theta)$. 
Each communication round consists of the following steps:

\paragraph{Client Forward Pass}
At the beginning of each round $t$, every client (in parallel) $k$ samples a mini-batch $B_k$ from its local dataset $\mathcal{D}_k$. The client performs a forward pass through its local model $\theta^C_k(t)$ using $B_k$, producing activations $a_k(t)$ at the cut layer. These activations and their corresponding ground truth labels $y_k(t)$ from $B_k$ are sent to the TS.

\paragraph{Training Server-side Computation}
The TS operation differs between SFL-V1 and SFL-V2. For SFL-V1, each client's activations $a_k(t)$ are processed by their corresponding server-side model $\theta^S_k(t)$. The TS computes predictions $\hat{y}_k(t)$ and the loss $\mathcal{L}(\hat{y}_k(t), y_k(t))$ for each client. Then, the TS computes and applies gradients $\nabla \theta^S_k(t)$ to update each server-side model. It computes gradients $\nabla a_k(t)$ at the cut layer, and sends $\nabla a_k(t)$ back to each client. 

For SFL-V2, client activations $a_k(t)$ are processed sequentially in a randomized manner through the shared server-side model $\theta^S(t)$. For each client's data, the TS computes predictions $\hat{y}_k(t)$ and loss $\mathcal{L}(\hat{y}_k(t), y_k(t))$, and updates $\theta^S(t)$ using gradient descent. Then, it computes gradients $\nabla a_k(t)$ at the cut layer, and sends $\nabla a_k(t)$ back to each client.

% For SFL-V1:
% \begin{itemize}
% \item Each client's activations $a_k(t)$ are processed by their corresponding server-side model $\theta^S_k(t)$.
% \item The TS computes predictions $\hat{y}_k(t)$ and the loss $\mathcal{L}(\hat{y}_k(t), y_k(t))$ for each client
% \item Through backpropagation, the TS:
%     \begin{itemize}
%         \item Computes and applies gradients $\nabla \theta^S_k(t)$ to update each server-side model
%         \item Computes gradients $\nabla a_k(t)$ at the cut layer
%         \item Sends $\nabla a_k(t)$ back to each client
%     \end{itemize}
% \end{itemize}

% For SFL-V2:
% \begin{itemize}
% \item Client activations $a_k(t)$ are processed sequentially through the shared server-side model $\theta^S(t)$.
% \item For each client's data, the TS:
%     \begin{itemize}
%         \item Computes predictions $\hat{y}_k(t)$ and loss $\mathcal{L}(\hat{y}_k(t), y_k(t))$
%         \item Updates $\theta^S(t)$ using gradient descent
%         \item Computes gradients $\nabla a_k(t)$ at the cut layer
%         \item Sends $\nabla a_k(t)$ back to each client
%     \end{itemize}
% \end{itemize}

\paragraph{Client Backward Pass}
After receiving $\nabla a_k(t)$, each client completes its backward pass to compute gradients $\nabla \theta^C_k(t)$. The client updates its model: $\theta^C_k(t+1) = \theta^C_k(t) - \eta^t \nabla \theta^C_k(t)$, where $\eta^t$ is the learning rate. This process repeats for $E$ local epochs before the next communication round.

\paragraph{Model Aggregation}
At the end of each round, the MSS performs FedAvg \cite{mcmahan2017communication} on the client-side models, where models are averaged weighted by their local dataset sizes to form $\theta^C(t+1)$. In SFL-V1, the MSS similarly aggregates its per-client server-side models $\theta^S_k(t)$ into a single model $\theta^S(t+1)$. In SFL-V2, the server-side model $\theta^S(t)$ does not perform aggregation.

\section{Theoretical Results}
In this section, we provide convergence analysis to understand the impact of cut layer selection. We start with some assumptions that are widely adopted in the distributed learning literature, e.g.,  \cite{han2024convergence,woodworth2020minibatch,wang2019adaptive}.
\begin{Assump} \label{assump: non-convexity}
Each client $k$'s loss function $F_k$ is non-convex. 
\end{Assump}
\begin{Assump} \label{assump: smoothness}
Each client $k$'s loss function $F_k$ is $S$-smooth. That is,  for all model parameters $\theta, \phi$,
\begin{equation}
F_k(\theta) \leq F_k(\phi)+\langle\nabla F_k(\phi), \theta-\phi\rangle+\frac{S}{2}\|\theta-\phi\|^2.
\end{equation}
\end{Assump}
% The smoothness assumption holds for many loss functions in, for example, logistic regression, softmax classifier, and $l_2$-norm regularized linear  regression \cite{li2020convergence}. 

\begin{Assump}\label{assump: stochastic gradient}
The stochastic gradients $g_k(\cdot)$ of $F_k\left(\cdot\right)$ 
are unbiased with the variance bounded by $\sigma_k^2$: 
\begin{equation}
    \mathbb{E}_{B_k\sim \mathcal{D}_n}\left[g_k\left(\theta, B_k\right)\right] = \nabla F_k\left(\theta\right),
\end{equation}
\begin{equation}\label{bounded-variance}
    \mathbb{E}_{B_k\sim \mathcal{D}_n}\left[\left\Vert g_k\left(\theta,B_k\right) - \nabla F_k\left(\theta\right) \right\Vert^2 \right] \leq \sigma_k^2.
\end{equation}
% Further, the expected squared norm is bounded by $G^2$:
% \begin{equation}
% \mathbb{E}_{\zeta_n\sim \mathcal{D}_n}\left\Vert \boldsymbol{g}_n\left(\theta, \zeta_n\right) \right\Vert^2 \leq G^2.
% \end{equation}
\end{Assump}
% Unbiasedness of the stochastic gradients holds for SGD (and its variants), as data samples are chosen with replacement.  
% The value of $\sigma_n$ measures the level of stochasticity. 

\begin{Assump}\label{assump: heterogeneity}
There exists an $\epsilon^2$ such that for any client $k$, 
\begin{equation}
\left\Vert \nabla F_k\left(\theta\right)-\nabla f\left(\theta\right)  \right\Vert^2 \le \epsilon^2.
\end{equation}
\end{Assump}
A larger $\epsilon^2$ indicates greater data heterogeneity. 
Next, we show that the convergence of SFL-V1 is invariant to the choice of cut layer $L_c$. 
% To this end, we first show that the convergence can be analyzed by client-side updates and server-side updates separately. 

% \begin{lemma}(Proposition 3.5 in \cite{han2024convergence})
% Let $\theta^*\triangleq\{\theta^{C*}; \theta^{S*}\}$ denote the optimal model,  and $\theta(T)=\{\theta^C(T); \theta^S(T)\}$ is the global model after $T$ rounds of SFL training. 
% Under Assumption \ref{assump: smoothness}, we have 
% \begin{equation}
% \begin{aligned}
% \mathbb{E}&\left[f(\theta(T))\right]- f(\theta^*)\\
% &\le\frac{S}{2}\left(\mathbb{E}||\theta^{S}(T)-\theta^{S*}||^2+\mathbb{E}||\theta^{C}(T)-\theta^{C*}||^2\right).
% \end{aligned}
% \end{equation}
% \end{lemma}

\begin{Prop}(Convergence Invariability to Cut Layer Selection in SFL-V1)\label{SFL-V1-invariability}
Let Assumptions \ref{assump: non-convexity}-\ref{assump: heterogeneity} hold, and let $\eta^t \leq \min\left\{ \frac{1}{16S\tau}, \frac{1}{8SK\tau\sum_{k=1}^K\alpha_k^2}\right\}$. Then, for any $L_c\in \{1, 2, \cdots, L-1\}$, the following inequality holds:
% then (\ref{bound-v1-nc-full}) holds.
\begin{equation}\label{bound-v1-nc-full}
\begin{aligned}
&\frac{1}{T}\sum_{t=0}^{T-1}\eta^t\mathbb{E}\left[\left\Vert\nabla f\left(\theta(t)\right)\right\Vert^2\right] \leq \frac{4}{T\tau}\left(f\left(\theta(0)\right)-f(\theta^{\ast}))\right) \\
&\hspace{10mm}+ \frac{16KS\tau}{T}\sum_{k=1}^K\alpha_k^2\left(\sigma_k^2+\epsilon^2\right)\sum_{t=0}^{T-1}\left(\eta^t\right)^2, 
\end{aligned}
% \begin{aligned}
% \frac{1}{T}\!\sum_{t=0}^{T-1}\!\eta^t\!\mathbb{E}\!\left[\!\left\Vert\nabla_{\theta}f\left(\theta^{t}\right)\right\Vert^2\!\right]\!\leq\!\frac{4}{T\!\tau_{\min}}\!\left(f\!\left(\!\theta^{0}\!\right)\!-\!f^\ast\right)\!+\!\frac{8NS(\tau^2\!+\!\tilde{\tau}^2)}{T\tau_{\rm min}}\!\sum_{n=1}^N\!a_n^2\!\left(\!\sigma_n^2+\epsilon^2\!\right)\!\sum_{t=0}^{T-1}\!\left(\eta^t\right)^2.
% \end{aligned}
\end{equation}
where $\theta^{\ast}$ is the optimal global model, $E$ is the number of epochs, and $\tau\triangleq\lceil \frac{E D_k}{B_k}\rceil$ denotes the number of model updates in one communication round. 
\end{Prop}

The proof of Proposition \ref{SFL-V1-invariability} is given in Appendix \ref{appendix: proof}. The convergence bound is affected by terms involving initial conditions $(f(\theta(0))) - f(\theta^{\ast}))$, data heterogeneity $\epsilon^2$, and gradient variance $\sigma_k^2$. A worse initial condition, a larger data heterogeneity, or a larger gradient variance will hurt the convergence, which is also observed in other distributed learning algorithms \cite{han2024incentivizing, huang2024federated}. Importantly, Proposition \ref{SFL-V1-invariability} holds for any cut layer selection. The key rationale is that SFL-V1 with different cut layers can be equivalently transformed into FedAvg with identical model updates.  This mathematical invariance to $L_c$ aligns with our empirical observations in Sec. \ref{sec:numerical} that SFL-V1's performance remains stable across different cut layer selections.
 
% \begin{proofsketch}
%     The convergence proof mainly follows \cite{han2024convergence}. The key difference is that we can 
% \end{proofsketch}

Note that the convergence proof of SFL-V2 would require a significantly different approach, and is left to future work. We will show in Sec. \ref{sec:numerical} that SFL-V2 outperforms SFL-V1 and FL under various experiment settings.

\section{Numerical Results}\label{sec:numerical}

We conduct experiments to validate our theoretical analysis and provide additional insights.

\begin{table}[t]
\caption{Dataset configurations showing architectures, client counts, and communication rounds.}
\label{tab:datasets_models_clients}
\centering
\begin{tabular}{l ccc}
\toprule
Dataset & Model & Clients & Rounds \\
\midrule
HAM10000 & ResNet-18 & 100 & 100 \\
CIFAR-10 & ResNet-18 & 100 & 200 \\
CIFAR-100 & ResNet-50 & 100 & 300 \\
Tiny ImageNet & ResNet-50 & 10 & 300 \\
\bottomrule
\end{tabular}
\end{table}

\begin{table}[t]
\caption{Training hyperparameters for experiments.}
\label{tab:hyperparameters}
\centering
\begin{tabular}{l ccc}
\toprule
Algorithm & Learning rate & Batch size & Optimizer \\
\midrule
SFL-V1 & 0.001 & 64 & Adam \\
SFL-V2 & 0.001 & 64 & Adam \\
FedAvg & 0.01 & 64 & SGD \\
\bottomrule
\end{tabular}
\end{table}
\subsection{Experimental Setup}

We conduct experiments on CIFAR-10 \cite{krizhevsky2009learning} and HAM10000 \cite{DVN/DBW86T_2018} with ResNet-18, and CIFAR-100 \cite{krizhevsky2009learning} and Tiny ImageNet \cite{Le_Yang_2015} with ResNet-50. 
The architectures use cut layer $L_c \in \{1, 2, 3, 4\}$, where $L_c = i$ denotes the network is cut following the $i$-th residual block. These cut points correspond to the boundaries between macro-residual blocks in ResNet architectures \cite{he2016deep}. We set the number of clients $K=100$, except for Tiny ImageNet, where $K=10$. The models, numbers of clients, and rounds of training used for each dataset are enumerated in Table \ref{tab:datasets_models_clients}. The local training epochs for each client were set as $E=5$. Hyperparameters for all algorithms are enumerated in Table \ref{tab:hyperparameters}. 

We consider both IID and non-IID scenarios. In our non-IID scenario, we use a label-based Dirichlet distribution \cite{hsu2019measuring, li2021model} with $\mu = 0.1$ to simulate highly imbalanced datasets across clients. Here, $\mu>0$ is a tunable parameter, where a smaller value of $\mu$ refers to a more imbalanced partition and $\mu=\infty$ corresponds to IID data partition. We exclude the HAM10000 dataset from this process because the dataset itself is highly label skewed. The results for IID and non-IID settings are presented in Tables \ref{tab:iid_results} and \ref{tab:noniid_results}, respectively. Bold values indicate the best performance for each dataset, while underlined values show the second-best results.

\subsubsection{Baseline Selection}
We use FedAvg as our FL baseline since innovations like proximal terms, control variates, and knowledge distillation could enhance both FL and SFL implementations. For instance, proximal terms could be added to client objectives in SFL to limit divergence, or control variates could be incorporated to correct for client drift. Our choice enables a fair comparison of the architectural differences between approaches without conflating algorithm-specific improvements.

\begin{table*}[ht]
\centering
\caption{Test accuracy (\%) comparison on IID data. Bold and underlined indicate best and second-best results respectively (mean ± std over 3 independent runs).}
\label{tab:iid_results}
\begin{tabular}{l ccc}
\toprule
& \multicolumn{1}{c}{ResNet-18} & \multicolumn{2}{c}{ResNet-50} \\
\cmidrule(lr){2-2} \cmidrule(lr){3-4}
Method & CIFAR10 & CIFAR100 & TinyImageNet \\
\midrule
FedAvg & 85.25 ± 0.07 & 48.48 ± 0.38 & 47.49 ± 0.07 \\
\addlinespace
SFL-V1 ($L_c=1$) & 86.34 ± 0.37 & 49.56 ± 0.18 & 48.73 ± 0.43 \\
SFL-V1 ($L_c=2$) & 86.12 ± 0.46 & \underline{50.28 ± 0.76} & 49.26 ± 0.78 \\
SFL-V1 ($L_c=3$) & 85.94 ± 0.36 & 49.29 ± 0.11 & 49.11 ± 0.70 \\
SFL-V1 ($L_c=4$) & 86.01 ± 0.45 & 46.59 ± 1.42 & 48.11 ± 0.27 \\
\addlinespace
SFL-V2 ($L_c=1$) & \textbf{92.30 ± 0.15} & \textbf{56.00 ± 0.69} & \textbf{54.77 ± 0.22} \\
SFL-V2 ($L_c=2$) & \underline{89.56 ± 0.15} & 45.93 ± 0.21 & \underline{50.86 ± 0.05} \\
SFL-V2 ($L_c=3$) & 87.57 ± 0.34 & 41.29 ± 0.26 & 44.15 ± 0.11 \\
SFL-V2 ($L_c=4$) & 86.11 ± 0.36 & 45.98 ± 0.24 & 43.85 ± 0.25 \\
\bottomrule
\end{tabular}
\end{table*}

\begin{table*}[ht]
\centering
\caption{Test accuracy (\%) comparison on non-IID data with Dirichlet $\mu=0.1$. Bold and underlined indicate best and second-best results respectively (mean ± std over independent 3 runs).}
\label{tab:noniid_results}
\begin{tabular}{l cccc}
\toprule
& \multicolumn{2}{c}{ResNet-18} & \multicolumn{2}{c}{ResNet-50} \\
\cmidrule(lr){2-3} \cmidrule(lr){4-5}
Method & HAM10000 & CIFAR10 & CIFAR100 & TinyImageNet \\
\midrule
FedAvg & 77.37 ± 0.35 & \underline{67.59 ± 2.52} & 42.60 ± 1.18 & \underline{28.33 ± 0.28} \\
\addlinespace
SFL-V1 ($L_c=1$) & 77.88 ± 0.13 & 66.38 ± 2.22 & 39.96 ± 0.84 & 12.84 ± 0.73 \\
SFL-V1 ($L_c=2$) & 78.18 ± 0.67 & 65.64 ± 1.49 & 39.82 ± 0.63 & 13.42 ± 0.60 \\
SFL-V1 ($L_c=3$) & 78.40 ± 0.66 & 65.36 ± 1.09 & 40.32 ± 0.88 & 13.66 ± 0.97 \\
SFL-V1 ($L_c=4$) & 78.10 ± 0.83 & 64.51 ± 1.84 & 41.24 ± 1.18 & 13.89 ± 1.00 \\
\addlinespace
SFL-V2 ($L_c=1$) & \textbf{80.58 ± 0.35} & 67.58 ± 6.28 & \textbf{52.38 ± 0.97} & \textbf{30.14 ± 8.58} \\
SFL-V2 ($L_c=2$) & \underline{79.56 ± 0.88} & 59.98 ± 11.99 & \underline{45.90 ± 1.08} & 22.55 ± 2.28 \\
SFL-V2 ($L_c=3$) & 79.26 ± 0.54 & 61.73 ± 4.04 & 42.42 ± 0.99 & 27.08 ± 1.48 \\
SFL-V2 ($L_c=4$) & 78.27 ± 0.65 & \textbf{69.45 ± 0.64} & 43.31 ± 1.17 & 26.42 ± 1.52 \\
\bottomrule
\end{tabular}
\end{table*}
\subsection{Impact of Cut Layer on SFL}

\subsubsection{SFL-V1: Robust Across Cut Layers}
SFL-V1 demonstrates stability across cut layers in both IID and non-IID settings, which is consistent with our analysis. For IID CIFAR10, the performance variation is minimal ($0.4\%$), while for non-IID CIFAR10, it's slightly higher but still modest ($1.87\%$). This stability stems from SFL-V1's algorithmic similarity to FedAvg, with minor variations potentially attributable to implementation-specific factors in PyTorch's computation graph. The robustness of SFL-V1 to cut layer selection enhances its deployability, particularly for resource-constrained edge devices and heterogeneous client environments.

\subsubsection{SFL-V2: Performance Affected by Cut Layer Selection} In contrast, SFL-V2 exhibits a strong dependency on cut layer selection. In three out of four datasets, SFL-V2 achieves peak performance with ($L_c=1$). The impact of cut layer becomes more pronounced in non-IID settings. For instance, on non-IID CIFAR100, SFL-V2's performance ranges from $42.42\%$ ($L_c=3$) to $52.38\%$ ($L_c=1$). While SFL-V2 generally achieves optimal performance with ($Lc=1$), we observe an interesting exception with non-IID CIFAR-10, where ($Lc=4$) performs best ($69.45\%$ vs $67.58\%$ at ($Lc=1$)). This outlier may be attributed to CIFAR-10's relatively simple feature space compared to CIFAR-100 or TinyImageNet. With only 10 classes, later layer features might be more discriminative for class separation under non-IID conditions. However, this pattern doesn't generalize across other datasets, suggesting that early cut layers remain the more robust choice for most applications.

The aforementioned behavior can be attributed to the fact that lower cut layers allow more data to pass through the training server, which more closely approximates centralized learning (CL), where all data is pooled and trained at a single location. Conceptually, $L_c=0$ would be equivalent to CL, while $L_c=L$ would reduce to FedAvg. Thus, the cut layer in SFL-V2 can serve as a crucial performance tuning parameter.

\subsection{SFL vs FL}
Our experiments demonstrate that SFL, particularly SFL-V2, often outperforms FL. In IID settings, SFL-V2 ($L_c=1$) consistently outperforms FedAvg across all datasets. For CIFAR10, SFL-V2 achieves $92.30\%$ accuracy compared to FedAvg's $85.25\%$, a significant $7.05\%$ improvement. The performance gap widens in non-IID scenarios. On CIFAR100, SFL-V2 ($L_c=1$) reaches $52.38\%$, while FedAvg achieves $42.60\%$, a substantial $9.78\%$ difference. This advantage is maintained with increasing dataset complexity: For non-IID TinyImageNet, SFL-V2 ($L_c=1$) achieves $30.14\%$ versus FedAvg's $28.33\%$.

Our results suggest that SFL-V2's superior performance stems from its unique architecture combining local and centralized computation. With an early cut layer, most of the model computation occurs at the training server, where the server-side model processes features from all clients without requiring weight aggregation. Unlike FedAvg, where models must reconcile entirely separate local training trajectories through weight averaging, SFL-V2's server-side model directly updates on client features through back propagation. This more direct form of learning appears to be more effective than weight aggregation, particularly for heterogeneous data distributions. The sequential processing of client features by a single shared server-side model, combined with client-side model aggregation, enables SFL-V2 to learn global patterns more efficiently than traditional FL approaches. This helps explain why SFL-V2 can surpass FedAvg, particularly when the cut layer is placed early in the network.

\section{Discussion}

\subsection{Addressing Non-IID Challenges}
While our experiments demonstrate SFL-V2's strong performance in non-IID settings, several important theoretical questions remain unexplored. First, we lack a theoretical framework explaining why SFL-V2's architecture provides advantages at certain cut layer selections. Understanding this could help optimize cut layer placement and guide architectural choices for different types of data heterogeneity. Additionally, while we've shown empirical improvements, we haven't established theoretical guarantees for the advantage of SFL-V2 over FedAvg.  A comprehensive theoretical analysis would need to account for the complex interplay between data heterogeneity, model architecture, and system constraints like client computation and communication resources.

\subsection{Privacy Considerations}
Although SFL maintains some level of privacy by keeping part of the model on client devices, the improved performance with lower cut layers in SFL-V2 raises important privacy considerations. As more layers are shifted to the server, there's an increased risk of potential privacy leakage through model inversions or inference attacks. Further research is needed to quantify these risks and develop mitigation strategies, possibly through the integration of differential privacy techniques or secure multi-party computation.

\section{Conclusion}
This paper presents an analysis of split federated learning, focusing on the impact of cut layer selection on model performance across various datasets and data distributions. Our study reveals that SFL-V1 exhibits stability across cut layers, while SFL-V2's performance is significantly influenced by the cut layer depth, with lower layers generally yielding better results. Notably, SFL-V2 outperforms FedAvg in both IID and non-IID settings. These findings highlight SFL's potential as a flexible and effective approach for privacy-preserving distributed learning, particularly in scenarios with heterogeneous data distributions.

For future work, it would be interesting to explore how to combine SFL with FL techniques handling data heterogeneity such as SCAFFOLD's control variates, FedProx's proximal terms, and MOON's contrastive learning approaches, using the unique availability of smashed data in SFL. Several challenges remain: developing a theoretical framework explaining SFL-V2's improved performance under non-IID conditions, quantifying privacy-performance trade-offs in cut layer selection, and validating these findings on real-world datasets. Additionally, developing adaptive algorithms that optimize these trade-offs based on specific application requirements, privacy constraints, and system limitations represents an important direction for future research.

\bibliography{main}

\begin{thebibliography}{51}
\providecommand{\natexlab}[1]{#1}

\bibitem[{Arivazhagan et~al.(2019)Arivazhagan, Aggarwal, Singh, and Choudhary}]{arivazhagan2019federated}
Arivazhagan, M.~G.; Aggarwal, V.; Singh, A.~K.; and Choudhary, S. 2019.
\newblock Federated learning with personalization layers.
\newblock \emph{arXiv preprint arXiv:1912.00818}.

\bibitem[{Brisimi et~al.(2018)Brisimi, Chen, Mela, Olshevsky, Paschalidis, and Shi}]{brisimi2018federated}
Brisimi, T.~S.; Chen, R.; Mela, T.; Olshevsky, A.; Paschalidis, I.~C.; and Shi, W. 2018.
\newblock Federated learning of predictive models from federated electronic health records.
\newblock \emph{International journal of medical informatics}, 112: 59--67.

\bibitem[{Caldas et~al.(2018)Caldas, Kone{\v{c}}ny, McMahan, and Talwalkar}]{caldas2018expanding}
Caldas, S.; Kone{\v{c}}ny, J.; McMahan, H.~B.; and Talwalkar, A. 2018.
\newblock Expanding the reach of federated learning by reducing client resource requirements.
\newblock \emph{arXiv preprint arXiv:1812.07210}.

\bibitem[{Chopra et~al.(2021)Chopra, Sahu, Singh, Java, Vepakomma, Sharma, and Raskar}]{chopra2021adasplit}
Chopra, A.; Sahu, S.~K.; Singh, A.; Java, A.; Vepakomma, P.; Sharma, V.; and Raskar, R. 2021.
\newblock Adasplit: Adaptive trade-offs for resource-constrained distributed deep learning.
\newblock \emph{arXiv preprint arXiv:2112.01637}.

\bibitem[{Duan et~al.(2022)Duan, Hu, Deng, and Lu}]{duan2022combined}
Duan, Q.; Hu, S.; Deng, R.; and Lu, Z. 2022.
\newblock Combined federated and split learning in edge computing for ubiquitous intelligence in internet of things: State-of-the-art and future directions.
\newblock \emph{Sensors}, 22(16): 5983.

\bibitem[{Gao et~al.(2024)Gao, Hu, Mashhadi, Wang, and Bennis}]{gao2024pipesfl}
Gao, Y.; Hu, B.; Mashhadi, M.~B.; Wang, W.; and Bennis, M. 2024.
\newblock PipeSFL: A Fine-Grained Parallelization Framework for Split Federated Learning on Heterogeneous Clients.
\newblock \emph{IEEE Transactions on Mobile Computing}.

\bibitem[{Gupta and Raskar(2018)}]{gupta2018distributed}
Gupta, O.; and Raskar, R. 2018.
\newblock Distributed learning of deep neural network over multiple agents.
\newblock \emph{Journal of Network and Computer Applications}, 116: 1--8.

\bibitem[{Hamer, Mohri, and Suresh(2020)}]{pmlr-v119-hamer20a}
Hamer, J.; Mohri, M.; and Suresh, A.~T. 2020.
\newblock {F}ed{B}oost: A Communication-Efficient Algorithm for Federated Learning.
\newblock In III, H.~D.; and Singh, A., eds., \emph{Proceedings of the 37th International Conference on Machine Learning}, volume 119 of \emph{Proceedings of Machine Learning Research}, 3973--3983. PMLR.

\bibitem[{Han et~al.(2024{\natexlab{a}})Han, Huang, Shi, Huang, and Liu}]{han2024incentivizing}
Han, P.; Huang, C.; Shi, X.; Huang, J.; and Liu, X. 2024{\natexlab{a}}.
\newblock Incentivizing Participation in SplitFed Learning: Convergence Analysis and Model Versioning.
\newblock In \emph{International Conference on Distributed Computing Systems}, 846--856. IEEE.

\bibitem[{Han et~al.(2024{\natexlab{b}})Han, Huang, Tian, Tang, and Liu}]{han2024convergence}
Han, P.; Huang, C.; Tian, G.; Tang, M.; and Liu, X. 2024{\natexlab{b}}.
\newblock Convergence Analysis of Split Federated Learning on Heterogeneous Data.
\newblock In \emph{The Thirty-eighth Annual Conference on Neural Information Processing Systems}.

\bibitem[{He et~al.(2016)He, Zhang, Ren, and Sun}]{he2016deep}
He, K.; Zhang, X.; Ren, S.; and Sun, J. 2016.
\newblock Deep residual learning for image recognition.
\newblock In \emph{Proceedings of the IEEE conference on computer vision and pattern recognition}, 770--778.

\bibitem[{Hsu, Qi, and Brown(2019)}]{hsu2019measuring}
Hsu, T.-M.~H.; Qi, H.; and Brown, M. 2019.
\newblock Measuring the effects of non-identical data distribution for federated visual classification.
\newblock \emph{arXiv preprint arXiv:1909.06335}.

\bibitem[{Huang, Dachille, and Liu(2024)}]{huang2024federated}
Huang, C.; Dachille, J.; and Liu, X. 2024.
\newblock When Federated Learning Meets Oligopoly Competition: Stability and Model Differentiation.
\newblock \emph{IEEE Internet of Things Journal}.

\bibitem[{Huang et~al.(2023)Huang, Zhang, Shao, Li, and Ji}]{huang2023roofsplit}
Huang, Y.; Zhang, H.; Shao, X.; Li, X.; and Ji, H. 2023.
\newblock RoofSplit: an edge computing framework with heterogeneous nodes collaboration considering optimal CNN model splitting.
\newblock \emph{Future Generation Computer Systems}, 140: 79--90.

\bibitem[{Karimireddy et~al.(2020)Karimireddy, Kale, Mohri, Reddi, Stich, and Suresh}]{karimireddy2020scaffold}
Karimireddy, S.~P.; Kale, S.; Mohri, M.; Reddi, S.; Stich, S.; and Suresh, A.~T. 2020.
\newblock Scaffold: Stochastic controlled averaging for federated learning.
\newblock In \emph{International conference on machine learning}, 5132--5143. PMLR.

\bibitem[{Khan et~al.(2022)Khan, Shejwalkar, Houmansadr, and Anwar}]{khan2022security}
Khan, M.~A.; Shejwalkar, V.; Houmansadr, A.; and Anwar, F.~M. 2022.
\newblock Security analysis of splitfed learning.
\newblock In \emph{Proceedings of the 20th ACM Conference on Embedded Networked Sensor Systems}, 987--993.

\bibitem[{Krizhevsky, Hinton et~al.(2009)}]{krizhevsky2009learning}
Krizhevsky, A.; Hinton, G.; et~al. 2009.
\newblock Learning multiple layers of features from tiny images.

\bibitem[{Le and Yang(2015)}]{Le_Yang_2015}
Le, Y.; and Yang, X. 2015.
\newblock Tiny imagenet visual recognition challenge.
\newblock \emph{CS 231N}, 7(7): 3.

\bibitem[{Lee et~al.(2024)Lee, Seif, Cho, and Poor}]{lee2024exploring}
Lee, J.; Seif, M.; Cho, J.; and Poor, H.~V. 2024.
\newblock Exploring the Privacy-Energy Consumption Tradeoff for Split Federated Learning.
\newblock \emph{IEEE Network}.

\bibitem[{Li et~al.(2022)Li, Rakin, Chen, He, Fan, and Chakrabarti}]{li2022ressfl}
Li, J.; Rakin, A.~S.; Chen, X.; He, Z.; Fan, D.; and Chakrabarti, C. 2022.
\newblock Ressfl: A resistance transfer framework for defending model inversion attack in split federated learning.
\newblock In \emph{Proceedings of the IEEE/CVF Conference on Computer Vision and Pattern Recognition}, 10194--10202.

\bibitem[{Li, He, and Song(2021)}]{li2021model}
Li, Q.; He, B.; and Song, D. 2021.
\newblock Model-contrastive federated learning.
\newblock In \emph{Proceedings of the IEEE/CVF conference on computer vision and pattern recognition}, 10713--10722.

\bibitem[{Li et~al.(2020)Li, Sahu, Zaheer, Sanjabi, Talwalkar, and Smith}]{li2020federated}
Li, T.; Sahu, A.~K.; Zaheer, M.; Sanjabi, M.; Talwalkar, A.; and Smith, V. 2020.
\newblock Federated optimization in heterogeneous networks.
\newblock \emph{Proceedings of Machine learning and systems}, 2: 429--450.

\bibitem[{Li and Lyu(2024)}]{li2024convergence}
Li, Y.; and Lyu, X. 2024.
\newblock Convergence analysis of sequential federated learning on heterogeneous data.
\newblock \emph{Advances in Neural Information Processing Systems}, 36.

\bibitem[{Li et~al.(2023)Li, Yan, Zhang, Gharibi, Yin, Jiang, and Malin}]{li2023split}
Li, Z.; Yan, C.; Zhang, X.; Gharibi, G.; Yin, Z.; Jiang, X.; and Malin, B.~A. 2023.
\newblock Split Learning for Distributed Collaborative Training of Deep Learning Models in Health Informatics.
\newblock In \emph{AMIA Annual Symposium Proceedings}, 1047. American Medical Informatics Association.

\bibitem[{Liao et~al.(2024)Liao, Xu, Xu, Wang, Yao, and Qiao}]{liao2024mergesfl}
Liao, Y.; Xu, Y.; Xu, H.; Wang, L.; Yao, Z.; and Qiao, C. 2024.
\newblock Mergesfl: Split federated learning with feature merging and batch size regulation.
\newblock In \emph{2024 IEEE 40th International Conference on Data Engineering (ICDE)}, 2054--2067. IEEE.

\bibitem[{Lin et~al.(2020)Lin, Kong, Stich, and Jaggi}]{lin2020ensemble}
Lin, T.; Kong, L.; Stich, S.~U.; and Jaggi, M. 2020.
\newblock Ensemble distillation for robust model fusion in federated learning.
\newblock \emph{Advances in neural information processing systems}, 33: 2351--2363.

\bibitem[{McMahan et~al.(2017)McMahan, Moore, Ramage, Hampson, and y~Arcas}]{mcmahan2017communication}
McMahan, B.; Moore, E.; Ramage, D.; Hampson, S.; and y~Arcas, B.~A. 2017.
\newblock Communication-efficient learning of deep networks from decentralized data.
\newblock In \emph{Artificial intelligence and statistics}, 1273--1282. PMLR.

\bibitem[{Minaee et~al.(2024)Minaee, Mikolov, Nikzad, Chenaghlu, Socher, Amatriain, and Gao}]{minaee2024large}
Minaee, S.; Mikolov, T.; Nikzad, N.; Chenaghlu, M.; Socher, R.; Amatriain, X.; and Gao, J. 2024.
\newblock Large language models: A survey.
\newblock \emph{arXiv preprint arXiv:2402.06196}.

\bibitem[{Mu and Shen(2023)}]{mu2023communication}
Mu, Y.; and Shen, C. 2023.
\newblock Communication and storage efficient federated split learning.
\newblock In \emph{IEEE International Conference on Communications}, 2976--2981. IEEE.

\bibitem[{Ninkovic et~al.(2024)Ninkovic, Vukobratovic, Miskovic, and Zennaro}]{ninkovic2024comsplit}
Ninkovic, V.; Vukobratovic, D.; Miskovic, D.; and Zennaro, M. 2024.
\newblock COMSPLIT: A Communication--Aware Split Learning Design for Heterogeneous IoT Platforms.
\newblock \emph{IEEE Internet of Things Journal}.

\bibitem[{Shin et~al.(2023)Shin, Ahn, Kang, and Kang}]{shin2023fedsplitx}
Shin, J.; Ahn, J.; Kang, H.; and Kang, J. 2023.
\newblock FedSplitX: Federated Split Learning for Computationally-Constrained Heterogeneous Clients.
\newblock \emph{arXiv preprint arXiv:2310.14579}.

\bibitem[{Shiranthika et~al.(2023)Shiranthika, Kafshgari, Saeedi, and Baji{\'c}}]{shiranthika2023splitfed}
Shiranthika, C.; Kafshgari, Z.~H.; Saeedi, P.; and Baji{\'c}, I.~V. 2023.
\newblock SplitFed resilience to packet loss: Where to split, that is the question.
\newblock In \emph{International Conference on Medical Image Computing and Computer-Assisted Intervention}, 367--377. Springer.

\bibitem[{Singh et~al.(2019)Singh, Vepakomma, Gupta, and Raskar}]{singh2019detailed}
Singh, A.; Vepakomma, P.; Gupta, O.; and Raskar, R. 2019.
\newblock Detailed comparison of communication efficiency of split learning and federated learning.
\newblock \emph{arXiv preprint arXiv:1909.09145}.

\bibitem[{Son et~al.(2024)Son, Kim, Chung, Huang, and Liu}]{son2024feduv}
Son, H.~M.; Kim, M.-H.; Chung, T.-M.; Huang, C.; and Liu, X. 2024.
\newblock FedUV: Uniformity and Variance for Heterogeneous Federated Learning.
\newblock In \emph{Proceedings of the IEEE/CVF Conference on Computer Vision and Pattern Recognition}, 5863--5872.

\bibitem[{Sun et~al.(2024)Sun, Yan, Jin, Zhao, and Chen}]{10371403}
Sun, W.; Yan, R.; Jin, R.; Zhao, R.; and Chen, Z. 2024.
\newblock FedAlign: Federated Model Alignment via Data-Free Knowledge Distillation for Machine Fault Diagnosis.
\newblock \emph{IEEE Transactions on Instrumentation and Measurement}, 73: 1--12.

\bibitem[{Tak and Cherkaoui(2021)}]{9220170}
Tak, A.; and Cherkaoui, S. 2021.
\newblock Federated Edge Learning: Design Issues and Challenges.
\newblock \emph{IEEE Network}, 35(2): 252--258.

\bibitem[{Thapa et~al.(2022)Thapa, Arachchige, Camtepe, and Sun}]{thapa2022splitfed}
Thapa, C.; Arachchige, P. C.~M.; Camtepe, S.; and Sun, L. 2022.
\newblock Splitfed: When federated learning meets split learning.
\newblock In \emph{Proceedings of the AAAI Conference on Artificial Intelligence}, volume~36, 8485--8493.

\bibitem[{Tschandl(2018)}]{DVN/DBW86T_2018}
Tschandl, P. 2018.
\newblock {The HAM10000 dataset, a large collection of multi-source dermatoscopic images of common pigmented skin lesions}.

\bibitem[{Vepakomma et~al.(2018)Vepakomma, Gupta, Swedish, and Raskar}]{vepakomma2018split}
Vepakomma, P.; Gupta, O.; Swedish, T.; and Raskar, R. 2018.
\newblock Split learning for health: Distributed deep learning without sharing raw patient data.
\newblock \emph{arXiv preprint arXiv:1812.00564}.

\bibitem[{Wang et~al.(2019)Wang, Tuor, Salonidis, Leung, Makaya, He, and Chan}]{wang2019adaptive}
Wang, S.; Tuor, T.; Salonidis, T.; Leung, K.~K.; Makaya, C.; He, T.; and Chan, K. 2019.
\newblock Adaptive federated learning in resource constrained edge computing systems.
\newblock \emph{IEEE journal on selected areas in communications}, 37(6): 1205--1221.

\bibitem[{Wang et~al.(2024)Wang, Lin, Liu, Zhang, and Liu}]{wang2024fedcst}
Wang, Z.; Lin, H.; Liu, Q.; Zhang, Y.; and Liu, X. 2024.
\newblock FedCST: Federated Learning on Heterogeneous Resource-constrained Devices Using Clustering and Split Training.
\newblock In \emph{2024 IEEE 24th International Conference on Software Quality, Reliability, and Security Companion (QRS-C)}, 786--792. IEEE.

\bibitem[{Williams, Waterman, and Patterson(2009)}]{williams2009roofline}
Williams, S.; Waterman, A.; and Patterson, D. 2009.
\newblock Roofline: an insightful visual performance model for multicore architectures.
\newblock \emph{Communications of the ACM}, 52(4): 65--76.

\bibitem[{Woodworth, Patel, and Srebro(2020)}]{woodworth2020minibatch}
Woodworth, B.~E.; Patel, K.~K.; and Srebro, N. 2020.
\newblock Minibatch vs local sgd for heterogeneous distributed learning.
\newblock \emph{Advances in Neural Information Processing Systems}, 33: 6281--6292.

\bibitem[{Wu et~al.(2023)Wu, Li, Qu, Zhou, Shen, Zhuang, Li, and Shi}]{wu2023split}
Wu, W.; Li, M.; Qu, K.; Zhou, C.; Shen, X.; Zhuang, W.; Li, X.; and Shi, W. 2023.
\newblock Split learning over wireless networks: Parallel design and resource management.
\newblock \emph{IEEE Journal on Selected Areas in Communications}, 41(4): 1051--1066.

\bibitem[{Xu et~al.(2023{\natexlab{a}})Xu, Li, Liu, Ling, and Wen}]{xu2023accelerating}
Xu, C.; Li, J.; Liu, Y.; Ling, Y.; and Wen, M. 2023{\natexlab{a}}.
\newblock Accelerating split federated learning over wireless communication networks.
\newblock \emph{IEEE Transactions on Wireless Communications}.

\bibitem[{Xu et~al.(2023{\natexlab{b}})Xu, Lyu, Dong, Lu, Wang, and Jin}]{xu2023splitgnn}
Xu, X.; Lyu, L.; Dong, Y.; Lu, Y.; Wang, W.; and Jin, H. 2023{\natexlab{b}}.
\newblock SplitGNN: Splitting GNN for Node Classification with Heterogeneous Attention.
\newblock \emph{arXiv preprint arXiv:2301.12885}.

\bibitem[{Yang et~al.(2022)Yang, Chen, Huangfu, Ran, Wang, Li, and Zhang}]{yang2022robust}
Yang, Z.; Chen, Y.; Huangfu, H.; Ran, M.; Wang, H.; Li, X.; and Zhang, Y. 2022.
\newblock Robust split federated learning for u-shaped medical image networks.
\newblock \emph{arXiv preprint arXiv:2212.06378}.

\bibitem[{Zhang et~al.(2022)Zhang, Qu, Singh, Kalpathy-Cramer, and Rubin}]{zhang2022splitavg}
Zhang, M.; Qu, L.; Singh, P.; Kalpathy-Cramer, J.; and Rubin, D.~L. 2022.
\newblock Splitavg: A heterogeneity-aware federated deep learning method for medical imaging.
\newblock \emph{IEEE Journal of Biomedical and Health Informatics}, 26(9): 4635--4644.

\bibitem[{Zhang et~al.(2023)Zhang, Pinto, Turina, Esposito, and Matta}]{zhang2023privacy}
Zhang, Z.; Pinto, A.; Turina, V.; Esposito, F.; and Matta, I. 2023.
\newblock Privacy and efficiency of communications in federated split learning.
\newblock \emph{IEEE Transactions on Big Data}, 9(5): 1380--1391.

\bibitem[{Zheng, Chen, and Lai(2024)}]{zheng2024ppsfl}
Zheng, J.; Chen, Y.; and Lai, Q. 2024.
\newblock PPSFL: Privacy-Preserving Split Federated Learning for heterogeneous data in edge-based Internet of Things.
\newblock \emph{Future Generation Computer Systems}, 156: 231--241.

\bibitem[{Zhu et~al.(2024)Zhu, Deng, Chen, Zhang, Fang, and Wong}]{zhu2024esfl}
Zhu, G.; Deng, Y.; Chen, X.; Zhang, H.; Fang, Y.; and Wong, T.~F. 2024.
\newblock ESFL: Efficient Split Federated Learning over Resource-Constrained Heterogeneous Wireless Devices.
\newblock \emph{IEEE Internet of Things Journal}.

\end{thebibliography}

% \newpage Newpage is not allowed
\onecolumn
\begin{appendix}

\section{Appendix: Algorithmic Framework}\label{appendix:algorithms}

\subsection{Notation}
We denote the number of clients as $K$ and the set of all clients as $\mathcal{K} = \{1,\ldots,K\}$. Each client $k$ has a local dataset $\mathcal{D}_k$ of size $|\mathcal{D}_k|$. The client-side model parameters for client $k$ are denoted as $\theta^C_k$, while the server-side model parameters are denoted as $\theta^S$ for SFL-V2 and $\theta^S_k$ for each client $k$ in SFL-V1. The cut layer index is denoted as $L_c$, and $B_k$ represents a mini-batch sampled from client $k$'s local dataset. We use $\eta$ as the learning rate, $E$ as the number of local epochs, and $T$ as the number of communication rounds. The activations at the cut layer from client $k$ are denoted as $a_k$, with $y_k$ representing the ground truth labels and $\hat{y}_k$ representing the predictions for client $k$'s data. The system involves two distinct server roles: the Training Server (TS), which maintains and updates server-side models, and the Model Synchronization Server (MSS), which handles model aggregation and distribution.

\begin{center}
\begin{minipage}{1.0\textwidth}
\subsection{SFL-V1 Algorithm}
\begin{algorithm}[H]
\caption{Split Federated Learning V1 (SFL-V1)}
\begin{algorithmic}[1]
\Require
    \State Clients $\mathcal{K} = \{1,\ldots,K\}$
    \State Communication rounds $T$
    \State Local epochs $E$
    \State Learning rate $\eta$
    \State Cut layer index $L_c$
    \State Local datasets $\{\mathcal{D}_k\}_{k=1}^K$
\Ensure
    \State Final client models $\{\theta^C_k(T)\}_{k=1}^K$
    \State Final server models $\{\theta^S_k(T)\}_{k=1}^K$

\State Initialize client models $\{\theta^C_k(0)\}_{k=1}^K$ and server models $\{\theta^S_k(0)\}_{k=1}^K$
\For{round $t = 0,\ldots,T-1$}
    \Statex \textbf{Client-side Parallel Processing:}
    \For{each client $k \in \mathcal{K}$ in parallel}
        \State Sample mini-batch $B_k \sim \mathcal{D}_k$
        \State Compute activations $a_k(t)$ through $\theta^C_k(t)$ up to $L_c$
        \State Send $(a_k(t), y_k(t))$ to Training Server
    \EndFor
    
    \Statex \textbf{Training Server Operations:}
    \For{each client $k \in \mathcal{K}$ in parallel}
        \State Compute predictions: $\hat{y}_k(t) = \theta^S_k(t)(a_k(t))$
        \State Calculate loss: $\mathcal{L}(\hat{y}_k(t), y_k(t))$
        \State Update server model: $\theta^S_k(t) \gets \theta^S_k(t) - \eta\nabla\theta^S_k(t)$
        \State Compute and send gradients $\nabla a_k(t)$ to client $k$
    \EndFor
    
    \Statex \textbf{Client Backward Pass:}
    \For{each client $k \in \mathcal{K}$ in parallel}
        \For{epoch $e = 1,\ldots,E$}
            \State Backpropagate using $\nabla a_k(t)$
            \State Update client model: $\theta^C_k(t) \gets \theta^C_k(t) - \eta\nabla\theta^C_k(t)$
        \EndFor
    \EndFor
    
    \Statex \textbf{Model Synchronization:}
    \State Aggregate client models: $\theta^C(t+1) \gets \sum_{k=1}^K \frac{|\mathcal{D}_k|}{\sum_i |\mathcal{D}_i|} \theta^C_k(t)$
    \State Aggregate server models: $\theta^S(t+1) \gets \sum_{k=1}^K \frac{|\mathcal{D}_k|}{\sum_i |\mathcal{D}_i|} \theta^S_k(t)$
    \State Broadcast aggregated models to all clients for next round:
    \For{each client $k \in \mathcal{K}$}
        \State Set client model: $\theta^C_k(t+1) \gets \theta^C(t+1)$
        \State Set server model: $\theta^S_k(t+1) \gets \theta^S(t+1)$
    \EndFor
\EndFor
\end{algorithmic}
\end{algorithm}
\end{minipage}
\end{center}

\begin{center}
\begin{minipage}{1.0\textwidth}
\subsection{SFL-V2 Algorithm}
\begin{algorithm}[H]
\caption{Split Federated Learning V2 (SFL-V2)}
\begin{algorithmic}[1]
\Require
    \State Clients $\mathcal{K} = \{1,\ldots,K\}$
    \State Communication rounds $T$
    \State Local epochs $E$
    \State Learning rate $\eta$
    \State Cut layer index $L_c$
    \State Local datasets $\{\mathcal{D}_k\}_{k=1}^K$
\Ensure
    \State Final client models $\{\theta^C_k(T)\}_{k=1}^K$
    \State Final server model $\theta^S(T)$

\State Initialize client models $\{\theta^C_k(0)\}_{k=1}^K$ and server model $\theta^S(0)$
\For{round $t = 0,\ldots,T-1$}
    \Statex \textbf{Client-side Parallel Processing:}
    \For{each client $k \in \mathcal{K}$ in parallel}
        \State Sample mini-batch $B_k \sim \mathcal{D}_k$
        \State Compute activations $a_k(t)$ through $\theta^C_k(t)$ up to $L_c$
        \State Send $(a_k(t), y_k(t))$ to Training Server
    \EndFor
    
    \Statex \textbf{Training Server Sequential Processing:}
    \State Generate random permutation $\pi$ of clients $\mathcal{K}$
    \For{$k \in \pi$}
        \State Compute predictions: $\hat{y}_k(t) = \theta^S(t)(a_k(t))$
        \State Calculate loss: $\mathcal{L}(\hat{y}_k(t), y_k(t))$
        \State Update server model: $\theta^S(t) \gets \theta^S(t) - \eta\nabla\theta^S(t)$
        \State Compute and send gradients $\nabla a_k(t)$ to client $k$
    \EndFor
    
    \Statex \textbf{Client Backward Pass:}
    \For{each client $k \in \mathcal{K}$ in parallel}
        \For{epoch $e = 1,\ldots,E$}
            \State Backpropagate using $\nabla a_k(t)$
            \State Update client model: $\theta^C_k(t) \gets \theta^C_k(t) - \eta\nabla\theta^C_k(t)$
        \EndFor
    \EndFor
    
    \Statex \textbf{Model Synchronization:}
    \State $\theta^C(t+1) \gets \sum_{k=1}^K \frac{|\mathcal{D}_k|}{\sum_i |\mathcal{D}_i|} \theta^C_k(t)$
    \State Broadcast aggregated model to all clients for next round:
    \For{each client $k \in \mathcal{K}$}
        \State $\theta^C_k(t+1) \gets \theta^C(t+1)$
    \EndFor
\EndFor
\end{algorithmic}
\end{algorithm}
\end{minipage}
\end{center}

\section{Proof of Proposition \ref{SFL-V1-invariability}}\label{appendix: proof}
The proof mainly follows \cite{han2024convergence}. The key difference is that we specify how the choice of cut layer affects the model updates at the training server and client side. For ease of presentation, we put the communication round index in the superscript, and the client/server notations ($c/s$) in the subscripts. 

\subsection{Preliminary}
We start with a few lemmas that facilitate the proof. 

\begin{lemma}\label{lem:multiple-local-training}[Multiple iterations of local training in each round] Under Assumptions \ref{assump: smoothness}-\ref{assump: heterogeneity}, if we let $\eta^t\leq\frac{1}{\sqrt{8}S\tau}$ and run client $k$'s local model for $\tau$ iteration continuously in any round $t$. Then, for any cut layer selection $L_c\in \{1, 2,\cdots, L-1\}$, we have 
    \begin{align}
    &\sum_{i=0}^{\tau-1}  \mathbb{E} \left[ \left\Vert  \theta^{t,i}_{k}- \theta^{t}\right\Vert^2\right]
\leq 2\tau^2\left(8\tau\left(\eta^t\right)^2 \sigma_k^2 +8\tau\left(\eta^t\right)^2\epsilon^2
         +8\tau\left(\eta^t\right)^2\left\Vert\nabla_{\theta}f\left(\theta^{t}\right)\right\Vert^2\right).
\end{align}
\end{lemma}
\begin{proof}
\begin{align}
    & \mathbb{E} \left[ \left\Vert \theta^{t,i}_{k}- \theta^{t}\right\Vert^2\right] \nonumber\\
       & \leq\mathbb{E} \left[ \left\Vert \theta^{t,i-1}_{k} - \eta^t \boldsymbol{g}_{k}^{t,i-1}- \theta^{t}\right\Vert^2\right] \nonumber\\
        & \leq \mathbb{E} \left[ \left\Vert \theta^{t,i-1}_{k} - \theta^{t} - \eta^t \left(\boldsymbol{g}_{k}^{t,i-1}-\nabla_{\theta}F_k\left(\theta^{t,i-1}_{k}\right)\right.\right.\right.\nonumber\\
        &\left.\left.\left.+\nabla_{\theta}F_k\left(\theta^{t,i-1}_{k}\right)-\nabla_{\theta}F_k\left(\theta^{t}\right)+\nabla_{\theta}F_k\left(\theta^{t}\right)-\nabla_{\theta}f\left(\theta^{t}\right)+\nabla_{\theta}f\left(\theta^{t}\right)\right)\right\Vert^2\right] \nonumber\\
        &\leq \left(1+\frac{1}{\tau}\right)\mathbb{E}\left[\left\Vert\theta^{t,i-1}_{k} - \theta^{t}  \right\Vert^2\right]+8\tau\mathbb{E}\left[\left\Vert\eta^t \left(\boldsymbol{g}_{k}^{t,i-1}-\nabla_{\theta}F_k\left(\theta^{t,i-1}_{k}\right)\right)\right\Vert^2\right]\nonumber\\
        &+8\tau\mathbb{E}\left[\left\Vert\eta^t \left(\nabla_{\theta}F_k\left(\theta^{t,i-1}_{k}\right)-\nabla_{\theta}F_k\left(\theta^{t}\right)\right)\right\Vert^2\right]
        +8\tau\mathbb{E}\left[\left\Vert\eta^t \left(\nabla_{\theta}F_k\left(\theta^{t}\right)-\nabla_{\theta}f\left(\theta^{t}\right)\right)\right\Vert^2\right]\nonumber\\
         &+8\tau\left\Vert\eta^t \nabla_{\theta}f\left(\theta^{t}\right)\right\Vert^2
       \nonumber\\
       &\leq \left(1+\frac{1}{\tau}\right)\mathbb{E}\left[\left\Vert\theta^{t,i-1}_{k} - \theta^{t}  \right\Vert^2\right]+8\tau\left(\eta^t\right)^2 \sigma_k^2 +8\tau\left(\eta^t\right)^2S^2\mathbb{E}\left[\left\Vert\theta^{t,i-1}_{k}-\theta^{t}\right\Vert^2\right]+8\tau\left(\eta^t\right)^2\epsilon^2\nonumber\\
         &+8\tau\left(\eta^t\right)^2\left\Vert\nabla_{\theta}f\left(\theta^{t}\right)\right\Vert^2
       \nonumber\\
         &\leq\!\left(\!1\!+\!\frac{1}{\tau}\!+\!8\tau\left(\!\eta^t\!\right)^2S^2\right)\mathbb{E}\!\left[\!\left\Vert\theta^{t,i-1}_{k}\!-\!\theta^{t}  \right\Vert^2\!\right]\!+\!8\tau\left(\eta^t\right)^2\!\sigma_k^2\!+\!8\tau\left(\!\eta^t\!\right)^2\epsilon^2\!+\!8\tau\left(\eta^t\right)^2\left\Vert\!\nabla_{\theta}f\left(\theta^{t}\!\right)\!\right\Vert^2\nonumber\\
          &\leq \left(1+\frac{2}{\tau}\right)\mathbb{E}\left[\left\Vert\theta^{t,i-1}_{k} - \theta^{t}  \right\Vert^2\right]+8\tau\left(\eta^t\right)^2 \sigma_k^2 +8\tau\left(\eta^t\right)^2\epsilon^2
         +8\tau\left(\eta^t\right)^2\left\Vert\nabla_{\theta}f\left(\theta^{t}\right)\right\Vert^2
\end{align}
where we have applied Assumptions \ref{assump: smoothness}-\ref{assump: heterogeneity}, $\left(X+Y\right)^2\leq\left(1+a\right)X^2+\left(1+\frac{1}{a}\right)Y^2$ for some positive $a$, and $\eta^t\leq\frac{1}{\sqrt{8}S\tau}$.

Let
\begin{align}
    &A_{t,i}:=\mathbb{E}\left[\left\Vert\theta^{t,i}_{k} - \theta^{t}  \right\Vert^2\right]\nonumber\\
    &B:=8\tau\left(\eta^t\right)^2 \sigma_k^2 +8\tau\left(\eta^t\right)^2\epsilon^2
         +8\tau\left(\eta^t\right)^2\left\Vert\nabla_{\theta}f\left(\theta^{t}\right)\right\Vert^2\nonumber\\
    &C:=1+\frac{2}{\tau}\nonumber
\end{align}

We have 
\begin{align}
    A_{t,i}\leq CA_{t,i-1}+B
\end{align}
We can show that
\begin{align}
   & A_{t,i}\leq C^iA_{t}+B\sum_{j=0}^{i-1}C^{j}\nonumber
\end{align}

Note that $A_{t}=\mathbb{E}\left[\left\Vert \theta^{t}-\theta^{t}\right\Vert^2\right]=0$. Accumulate the above for $\tau$ iterations, we have
\begin{align}
    &\sum_{i=0}^{\tau-1}  \mathbb{E} \left[ \left\Vert  \theta^{t,i}_{k}- \theta^{t}\right\Vert^2\right] =\sum_{i=0}^{\tau-1} B\sum_{j=0}^{i-1}C^{j} \nonumber\\
   & \leq 2\tau^2B \nonumber\\
    & \leq 2\tau^2\left(8\tau\left(\eta^t\right)^2 \sigma_k^2 +8\tau\left(\eta^t\right)^2\epsilon^2
         +8\tau\left(\eta^t\right)^2\left\Vert\nabla_{\theta}f\left(\theta^{t}\right)\right\Vert^2\right)
\end{align}
where we use $\sum_{i=0}^{N-1}x^i=\frac{x^N-1}{x-1}$ and $ (1+\frac{n}{x})^x\leq e^n$.
Therefore, we complete the proof.
\end{proof}

\begin{lemma}\label{lem:multiple-local-training-grad}[Multiple iterations of local gradient accumulation in each round] Under Assumption \ref{assump: smoothness}-\ref{assump: heterogeneity}, if we let $\eta^t\leq\frac{1}{2S\tau}$ and run client $k$'s local model for $\tau$ iterations continuously in any round $t$. Then, for any cut layer selection $L_c\in \{1, 2,\cdots, L-1\}$, we have  
\begin{align}
    \sum_{i=0}^{\tau-1}\mathbb{E}\left[\left\Vert \boldsymbol{g}^{t,i}_{k}-\boldsymbol{g}_{k}^t\right\Vert^2\right] \leq 8\tau^3\left(\eta^t\right)^2S^2\left(\left\Vert \nabla_{\theta}F_k\left(\theta^{t}\right)\right\Vert^2+\sigma_k^2\right). 
\end{align}
\end{lemma}
\begin{proof}
\begin{align}
    &\mathbb{E}\left[\left\Vert \boldsymbol{g}^{t,i}_{k}-\boldsymbol{g}_{k}^t\right\Vert^2\right]\nonumber\\
    &\leq \mathbb{E}\left[\left\Vert \boldsymbol{g}^{t,i}_{k}-\boldsymbol{g}_{k}^{t,i-1}+\boldsymbol{g}_{k}^{t,i-1}-\boldsymbol{g}_{k}^t\right\Vert^2\right]\nonumber\\
    & \leq \left(1+\tau\right)\mathbb{E}\left[\left\Vert \boldsymbol{g}^{t,i}_{k}-\boldsymbol{g}_{k}^{t,i-1}\right\Vert^2\right]+ 
    \left(1+\frac{1}{\tau}\right)\mathbb{E}\left[\left\Vert \boldsymbol{g}_{k}^{t,i-1}-\boldsymbol{g}_{k}^t\right\Vert^2\right] \nonumber\\
    & \leq \left(1+\tau\right)S^2\mathbb{E}\left[\left\Vert \theta^{t,i}_{k}-\theta^{t,i-1}_{k}\right\Vert^2\right]+ 
    \left(1+\frac{1}{\tau}\right)\mathbb{E}\left[\left\Vert \boldsymbol{g}_{k}^{t,i-1}-\boldsymbol{g}_{k}^t\right\Vert^2\right] \nonumber\\
    & \leq \left(1+\tau\right)\left(\eta^t\right)^2S^2\mathbb{E}\left[\left\Vert \boldsymbol{g}_{k}^{t,i-1}\right\Vert^2\right]+ 
    \left(1+\frac{1}{\tau}\right)\mathbb{E}\left[\left\Vert \boldsymbol{g}_{k}^{t,i-1}-\boldsymbol{g}_{k}^t\right\Vert^2\right] \nonumber\\
    & \leq \left(1+\tau\right)\left(\eta^t\right)^2S^2\mathbb{E}\left[\left\Vert \boldsymbol{g}_{k}^{t,i-1}-\boldsymbol{g}_{k}^t+\boldsymbol{g}_{k}^t\right\Vert^2\right]+ 
    \left(1+\frac{1}{\tau}\right)\mathbb{E}\left[\left\Vert \boldsymbol{g}_{k}^{t,i-1}-\boldsymbol{g}_{k}^t\right\Vert^2\right] \nonumber \\
    & \leq 2\left(1+\tau\right)\left(\eta^t\right)^2S^2\mathbb{E}\left[\left\Vert \boldsymbol{g}_{k}^{t,i-1}-\boldsymbol{g}_{k}^t\right\Vert^2\right] + 2\left(1+\tau\right)\left(\eta^t\right)^2S^2\mathbb{E}\left[\left\Vert \boldsymbol{g}_{k}^t\right\Vert^2\right] \nonumber\\ 
    &+ \left(1+\frac{1}{\tau}\right)\mathbb{E}\left[\left\Vert \boldsymbol{g}_{k}^{t,i-1}-\boldsymbol{g}_{k}^t\right\Vert^2\right] \nonumber \\
    & \leq \left(1+\frac{2}{\tau}\right)\mathbb{E}\left[\left\Vert \boldsymbol{g}_{k}^{t,i-1}-\boldsymbol{g}_{k}^t\right\Vert^2\right] + 2\left(1+\tau\right)\left(\eta^t\right)^2S^2\mathbb{E}\left[\left\Vert \boldsymbol{g}_{k}^t\right\Vert^2\right] .
\end{align}

We define the following notation for simplicity:
\begin{align}
    &A_{t,i}:= \mathbb{E}\left[\left\Vert \boldsymbol{g}^{t,i}_{k}-\boldsymbol{g}_{k}^t\right\Vert^2\right]\\
    &B:=2\left(1+\tau\right)\left(\eta^t\right)^2S^2\mathbb{E}\left[\left\Vert \boldsymbol{g}_{k}^t\right\Vert^2\right]\\
    &C:=\left(1+\frac{2}{\tau}\right)
\end{align}

We have 
\begin{align}
    A_{t,i}\leq CA_{t,i-1}+B
\end{align}
We can show that
\begin{align}
  & A_{t,i}\leq C^iA_{t}+B\sum_{j=0}^{i-1}C^{j}\nonumber
\end{align}

Note that $A_{t}=\mathbb{E}\left[\left\Vert \boldsymbol{g}_{k}^t-\boldsymbol{g}_{k}^t\right\Vert^2\right]=0$. For the second part, we have
\begin{align}
    &\sum_{i=0}^{\tau-1}\mathbb{E}\left[\left\Vert \boldsymbol{g}^{t,i}_{k}-\boldsymbol{g}_{k}^t\right\Vert^2\right] =\sum_{i=0}^{\tau-1} B\sum_{j=0}^{i-1}C^{j}  \leq 2\tau^2B \nonumber\\
    & \leq 4\tau^2\left(1+\tau\right)\left(\eta^t\right)^2S^2\mathbb{E}\left[\left\Vert \boldsymbol{g}_{k}^t\right\Vert^2\right]\nonumber\\
    & \leq 8\tau^3\left(\eta^t\right)^2S^2\mathbb{E}\left[\left\Vert \boldsymbol{g}_{k}^t\right\Vert^2\right]\nonumber\\
    & \leq 8\tau^3\left(\eta^t\right)^2S^2\left(\left\Vert \nabla_{\theta}F_k\left(\theta^{t}\right)\right\Vert^2+\sigma_k^2\right) .
\end{align}
\end{proof}

\subsection{Main server's model update}
We first analyze one-round model update at the main server side. For any cut layer selection $L_c\in \{1, 2,\cdots, L-1\}$, we have 
\begin{align}
    &\mathbb{E} \left[\left\langle\nabla_{\theta_s} f\left(\theta^{t}\right), \theta^{t+1}_s-\theta^{t}_s\right\rangle\right]\nonumber\\
    &\leq \mathbb{E} \left[\left\langle\nabla_{\theta_s} f\left(\theta^{t}\right), \theta^{t+1}_s-\theta^{t}_s +\eta^t \tau \nabla_{\theta_s} f\left(\theta^{t}\right)- \eta^t \tau \nabla_{\theta_s} f\left(\theta^{t}\right)\right\rangle\right]\nonumber\\
     &\leq \mathbb{E} \left[\left\langle\nabla_{\theta_s} f\left(\theta^{t}\right), \theta^{t+1}_s-\theta^{t}_s +\eta^t \tau \nabla_{\theta_s} f\left(\theta^{t}\right)\right\rangle -\left\langle f\left(\theta^{t}_s\right), \eta^t \tau \nabla_{\theta_s} f\left(\theta^{t}\right)\right\rangle\right]\nonumber\\
     &\leq \left\langle\nabla_{\theta_s} f\left(\theta^{t}\right), \mathbb{E} \left[-\eta^t\sum_{k=1}^K \sum_{i=0}^{\tau-1}\alpha_k\boldsymbol{g}_{s,k}^{t,i}\right]+\eta^t \tau \nabla_{\theta_s} f\left(\theta^{t}\right)\right\rangle -\eta^t \tau \left\Vert\nabla_{\theta_s} f\left(\theta^{t}\right)\right\Vert^2 \nonumber\\
     &\leq \left\langle\nabla_{\theta_s} f\left(\theta^{t}\right), \mathbb{E} \left[- \eta^t\sum_{k=1}^K \sum_{i=0}^{\tau-1}\alpha_k\nabla_{\theta_s} F_k\left(\left\{\theta^{t,i}_{c,k},\theta^{t,i}_{s,k}\right\}\right)\right]+\eta^t \tau \nabla_{\theta_s} f\left(\theta^{t}\right)\right\rangle -\eta^t \tau \left\Vert\nabla_{\theta_s} f\left(\theta^{t}\right)\right\Vert^2 \nonumber\\
     &\leq \left\langle\nabla_{\theta_s} f\left(\theta^{t}\right), \mathbb{E} \left[- \eta^t\sum_{k=1}^K \sum_{i=0}^{\tau-1} \alpha_k\nabla_{\theta_s} F_k\left(\left\{\theta^{t,i}_{c,k},\theta^{t,i}_{s,k}\right\}\right)+\eta^t\sum_{k=1}^K \sum_{i=0}^{\tau-1} \alpha_k\nabla_{\theta_s} F_k\left(\theta^{t}\right) \right]\right\rangle \nonumber\\
     &\quad -\eta^t \tau \left\Vert\nabla_{\theta_s} f\left(\theta^{t}\right)\right\Vert^2 \nonumber\\
      &\leq  \eta^t\tau\left\langle \nabla_{\theta_s} f\left(\theta^{t}\right),\mathbb{E} \left[ - \frac{1}{\tau}\sum_{k=1}^K \sum_{i=0}^{\tau-1}\alpha_k\nabla_{\theta_s} F_k\left(\left\{\theta^{t,i}_{c,k},\theta^{t,i}_{s,k}\right\}\right)+ \frac{1}{\tau}\sum_{k=1}^K \sum_{i=0}^{\tau-1} \alpha_k\nabla_{\theta_s} F_k\left(\theta^{t}\right) \right] \right\rangle
      \nonumber\\
      &\quad -\eta^t \tau \left\Vert\nabla_{\theta_s} f\left(\theta^{t}\right)\right\Vert^2 \nonumber\\
      &\leq \frac{\eta^t\tau}{2} \left\Vert\nabla_{\theta_s} f\left(\theta^{t}\right)\right\Vert^2 +\frac{\eta^t}{2\tau}\mathbb{E} \left[ \left\Vert \sum_{k=1}^K \sum_{i=0}^{\tau-1} \alpha_k\nabla_{\theta_s} F_k\left(\left\{\theta^{t,i}_{c,k},\theta^{t,i}_{s,k}\right\}\right)- \sum_{k=1}^K \sum_{i=0}^{\tau-1} \alpha_k\nabla_{\theta_s} F_k\left(\theta^{t}\right) \right\Vert^2\right]
      \nonumber\\
      &\quad -\eta^t \tau \left\Vert\nabla_{\theta_s} f\left(\theta^{t}\right)\right\Vert^2 \nonumber\\
       &\leq -\frac{\eta^t\tau}{2} \left\Vert\nabla_{\theta_s} f\left(\theta^{t}\right)\right\Vert^2 +\frac{K\eta^t}{2 \tau}\sum_{k=1}^K\alpha_k^2\mathbb{E} \left[ \left\Vert \sum_{i=0}^{\tau-1} \left(\nabla_{\theta_s} F_k\left(\left\{\theta^{t,i}_{c,k},\theta^{t,i}_{s,k}\right\}\right)- \nabla_{\theta_s} F_k\left(\theta^{t}\right) \right)\right\Vert^2\right]\nonumber\\
       &\leq -\frac{\eta^t\tau}{2} \left\Vert\nabla_{\theta_s} f\left(\theta^{t}\right)\right\Vert^2 +\frac{K\eta^t S^2}{2}\sum_{k=1}^K \alpha_k^2\sum_{i=0}^{\tau-1} \mathbb{E} \left[ \left\Vert \theta^{t,i}_{s,k}- \theta^{t}_s\right\Vert^2\right],\label{eq:server-non-full-1.1}
\end{align}
where we have used the fact that $\nabla_{\theta_s} f\left(\theta^{t}\right)=\sum_{k=1}^K \alpha_k\nabla_{\theta_s} F_k\left(\theta^{t}\right)$, and $\left\langle a,b\right\rangle\leq \frac{a^2+b^2}{2}$.

By Lemma \ref{lem:multiple-local-training} with $\eta^t\leq\frac{1}{\sqrt{8}S\tau}$, we have
\begin{align}
    &\sum_{i=0}^{\tau-1}  \mathbb{E} \left[ \left\Vert  \theta^{t,i}_{s,k}- \theta^{t}_s\right\Vert^2\right]\leq 2\tau^2\left(8\tau\left(\eta^t\right)^2 \sigma_k^2 +8\tau\left(\eta^t\right)^2\epsilon^2
         +8\tau\left(\eta^t\right)^2\left\Vert\nabla_{\theta_s}f\left(\theta^{t}_s\right)\right\Vert^2\right).
\end{align}

Thus, \eqref{eq:server-non-full-1.1} becomes
\begin{align}
    &\mathbb{E} \left[\left\langle\nabla_{\theta_s} f\left(\theta^{t}\right), \theta^{t+1}_s-\theta^{t}_s\right\rangle\right]\nonumber\\        
    &\leq\!-\frac{\eta^t\tilde{\tau}}{2} \left\Vert\!\nabla_{\theta_s} f\left(\theta^{t}\right)\!\right\Vert^2\!+\!\frac{N\eta^t S^2}{2}\!\sum_{n=1}^N\!a_n^22\tilde{\tau}^2\left(\!8\tilde{\tau}\left(\eta^t\right)^2 \sigma_n^2\!+\!8\tilde{\tau}\left(\eta^t\right)^2\epsilon^2\!+\!8\tilde{\tau}\left(\eta^t\right)^2\!\left\Vert\nabla_{\theta_s}\!f\left(\theta^{t}_s\right)\right\Vert^2\!\right)\nonumber\\
       &\leq \left(-\frac{\eta^t\tilde{\tau}}{2} +8N\left(\eta^t\right)^3\tilde{\tau}^3 S^2\sum_{n=1}^N a_n^2\right)\left\Vert\nabla_{\theta_s} f\left(\theta^{t}\right)\right\Vert^2 +8N\eta^t S^2\tilde{\tau}^3\sum_{n=1}^N a_n^2\left(\eta^t\right)^2\left( \sigma_n^2 +\epsilon^2
         \right).\label{eq:server-non-full-1.5}
\end{align}

Furthermore, we have
\begin{align}
    &\frac{S}{2}\mathbb{E}\left[\left\|  \theta^{t+1}_{s}- \theta^{t}_{s} \right\|^2 \right]\nonumber\\
    &=\frac{SK\left(\eta^t\right)^2}{2}\sum_{k=1}^K\mathbb{E}\left[\left\|  \sum_{i=0}^{\tau-1} \alpha_k\boldsymbol{g}^{t,i}_{s,k}\right\|^2 \right]\nonumber\\
    &\leq\frac{SK\left(\eta^t\right)^2}{2}\sum_{k=1}^K\alpha_k^2\mathbb{E}\left[\left\|\sum_{i=0}^{\tau-1} \boldsymbol{g}^{t,i}_{s,k}\right\|^2 \right]\nonumber\\
      &\leq\frac{SK\left(\eta^t\right)^2 \tau}{2}\sum_{k=1}^K\alpha_k^2\sum_{i=0}^{\tau-1} \mathbb{E}\left[\left\|\boldsymbol{g}^{t,i}_{s,k}\right\|^2 \right]\nonumber\\
      &\leq\frac{SK\left(\eta^t\right)^2 \tau}{2}\sum_{k=1}^K\alpha_k^2\sum_{i=0}^{\tau-1} \mathbb{E}\left[\left\|\boldsymbol{g}^{t,i}_{s,k}-\boldsymbol{g}_{s,k}^t+\boldsymbol{g}_{s,k}^t\right\|^2 \right]\nonumber\\
       &\leq \frac{SK\left(\eta^t\right)^2 \tau}{2}\sum_{k=1}^K\alpha_k^2\sum_{i=0}^{\tau-1} \left(\mathbb{E}\left[\left\|\boldsymbol{g}^{t,i}_{s,k}-\boldsymbol{g}_{s,k}^t\right\|^2\right]+\mathbb{E}\left[\left\|\boldsymbol{g}_{s,k}^t\right\|^2 \right]\right)\nonumber\\
       &\leq \frac{SK\left(\eta^t\right)^2 \tau}{2}\sum_{k=1}^K\alpha_k^2\sum_{i=0}^{\tau-1} \left(\mathbb{E}\left[\left\|\boldsymbol{g}^{t,i}_{s,k}-\boldsymbol{g}_{s,k}^t\right\|^2\right]+\mathbb{E}\left[\left\|\nabla_{\theta_s}F_k\left(\theta^{t}\right)\right\|^2+\sigma_k^2 \right]\right), \label{eq:server-non-full-1.3}
\end{align}
where the last line uses Assumption \ref{assump: smoothness}-\ref{assump: heterogeneity} and $\mathbb{E}\left[\|\mathbf{z}\|^2\right]=\|\mathbb{E}[\mathbf{z}]\|^2+\mathbb{E}[\| \mathbf{z}-\mathbb{E}[\mathbf{z}] \|^2]$ for any random variable $\mathbf{z}$.

By Lemma \ref{lem:multiple-local-training-grad} with $\eta^t\leq \frac{1}{2S\tau}$, we have 
\begin{align}
    \sum_{i=0}^{\tau-1}\mathbb{E}\left[\left\Vert \boldsymbol{g}^{t,i}_{s,k}-\boldsymbol{g}_{s,k}^t\right\Vert^2\right] \leq 8\tau^3\left(\eta^t\right)^2S^2\left(\left\Vert \nabla_{\theta_s}F_k\left(\theta^{t}\right)\right\Vert^2+\sigma_k^2\right). 
\end{align}
Thus,
\begin{align}
    &\frac{S}{2}\mathbb{E}\left[\left\|  \theta^{t+1}_{s}- \theta^{t}_{s} \right\|^2 \right]\nonumber\\
     &\leq\frac{SK\left(\eta^t\right)^2 \tau}{2}\sum_{k=1}^K\alpha_k^2\left(8\tau^3\left(\eta^t\right)^2S^2\left(\left\Vert \nabla_{\theta_s}F_k\left(\theta^{t}\right)\right\Vert^2+\sigma_k^2\right) 
+\tau\mathbb{E}\left[\left\|\nabla_{\theta_s}F_k\left(\theta^{t}\right)\right\|^2+\sigma_k^2 \right] \right)\nonumber\\
     &\leq\frac{SK\left(\eta^t\right)^2 \tau}{2}\sum_{k=1}^K\alpha_k^2\left(\tau+8\tau^3\left(\eta^t\right)^2S^2\right)\left(\left\Vert \nabla_{\theta_s}F_k\left(\theta^{t}\right)\right\Vert^2+\sigma_k^2\right)\nonumber\\
    &\leq\frac{SK\left(\eta^t\right)^2 \tau}{2}\sum_{k=1}^K\alpha_k^2\left(\tau+8\tau^3\left(\eta^t\right)^2S^2\right)\left(\left\Vert \nabla_{\theta_s}F_k\left(\theta^{t}\right)-\nabla_{\theta_s}f\left(\theta^{t}\right)+\nabla_{\theta_s}f\left(\theta^{t}\right)\right\Vert^2+\sigma_k^2\right)\nonumber\\
     &\leq\frac{SK\left(\eta^t\right)^2 \tau}{2}\sum_{k=1}^K\alpha_k^2\left(\tau+8\tau^3\left(\eta^t\right)^2S^2\right)\left(2\left\Vert \nabla_{\theta_s}f\left(\theta^{t}\right)\right\Vert^2+2\epsilon^2+\sigma_k^2\right).
     \label{eq:server-non-full-1.4}
\end{align}

\subsection{Clients' model update} 
The analysis of the client-side model update is similar to the server. For any cut layer selection $L_c\in \{1, 2,\cdots, L-1\}$, we have 
\begin{align}
    &\mathbb{E} \left[\left\langle\nabla_{\theta_c} f\left(\theta^{t}\right), \theta^{t+1}_c-\theta^{t}_c\right\rangle\right]\nonumber\\        
    &\leq \left(-\frac{\eta^t\tau}{2} +8K\left(\eta^t\right)^3\tau^3 S^2\sum_{k=1}^K \alpha_k^2\right)\left\Vert\nabla_{\theta_c} f\left(\theta^{t}\right)\right\Vert^2 +8K\eta^t S^2\tau^3\sum_{k=1}^K \alpha_k^2\left(\eta^t\right)^2\left( \sigma_k^2 +\epsilon^2
         \right).\label{eq:client-non-full-1.5}
\end{align}

For $\eta^t\leq \frac{1}{2S\tau}$,
\begin{align}
    &\frac{S}{2}\mathbb{E}\left[\left\|  \theta^{t+1}_{c}- \theta^{t}_{c} \right\|^2 \right]\nonumber\\
     &\leq\frac{SK\left(\eta^t\right)^2 \tau}{2}\sum_{k=1}^K\alpha_k^2\left(\tau+8\tau^3\left(\eta^t\right)^2S^2\right)\left(2\left\Vert \nabla_{\theta_c}f\left(\theta^{t}\right)\right\Vert^2+2\epsilon^2+\sigma_k^2\right).
     \label{eq:client-non-full-1.4}
\end{align}

\subsection{Superposition of main server and clients}
% Define $\tau_{\rm min}\triangleq \min\{\tau, \tilde{\tau}\}, \quad  \tau_{\rm max}\triangleq \max\{\tau, \tilde{\tau}\}$.
For any cut layer selection $L_c\in \{1, 2,\cdots, L-1\}$, we have 
\begin{align}
&\mathbb{E}\left[f\left(\theta^{t+1}\right)\right]-f\left(\theta^{t}\right)\nonumber\\
& \leq \mathbb{E}\left[\left\langle\nabla_{\theta_c} f\left(\theta^{t}\right), \theta^{t+1}_c-\theta^{t}_c\right\rangle\right]+\frac{S}{2}\mathbb{E}\left[\left\|\theta^{t+1}_c-\theta^{t}_c\right\|^2\right]+\mathbb{E}\left[\left\langle\nabla_{\theta_s} f\left(\theta^{t}\right), \theta^{t+1}_s-\theta^{t}_s\right\rangle\right]\nonumber\\
&+\frac{S}{2}\mathbb{E}\left[\left\|\theta^{t+1}_s-\theta^{t}_s\right\|^2\right]\nonumber\\
& \leq\left(-\frac{\eta^t\tau}{2}+8K\left(\eta^t\right)^3S^2\tau^3\sum_{k=1}^K\alpha_k^2+SK\left(\eta^t\right)^2\tau\sum_{k=1}^K\alpha_k^2\left(\tau+8\tau^3\left(\eta^t\right)^2S^2\right)\right)\left\Vert\nabla_{\theta}f\left(\theta^{t}\right)\right\Vert^2 \nonumber\\
   &+8K\eta^t S^2\tau^3\sum_{k=1}^K\alpha_k^2\left(\eta^t\right)^2\left( \sigma_k^2 +\epsilon^2 \right) \nonumber\\
   &+\frac{1}{2}SK\left(\eta^t\right)^2\tau\left(\tau+8\tau^3\left(\eta^t\right)^2S^2\right)\sum_{k=1}^K\alpha_k^2\left(2\epsilon^2+\sigma_k^2\right)\nonumber\\
   &\leq\left(-\frac{\eta^t\tau}{2}+SK\left(\eta^t\right)^2\tau^2\sum_{k=1}^K\alpha_k^2+8K\left(\eta^t\right)^3S^2\tau^3\sum_{k=1}^K\alpha_k^2+8S^3K\left(\eta^t\right)^4\tau^4\sum_{k=1}^K\alpha_k^2\right)\left\Vert\nabla_{\theta}f\left(\theta^{t}\right)\right\Vert^2 \nonumber\\
   &+8K\left(\eta^t\right)^3 S^2\tau^3\sum_{k=1}^K\alpha_k^2\sigma_k^2 +8K\left(\eta^t\right)^3 S^2\tau^3\epsilon^2\sum_{k=1}^K\alpha_k^2 \nonumber\\
   &+SK\left(\eta^t\right)^2\tau^2\epsilon^2\sum_{k=1}^K\alpha_k^2
   +\frac{1}{2}SK\left(\eta^t\right)^2\tau^2\sum_{k=1}^K\alpha_k^2\sigma_k^2+8KS^3\left(\eta^t\right)^4\tau^4\epsilon^2\sum_{k=1}^K\alpha_k^2
   +4KS^3\left(\eta^t\right)^4\tau^4\sum_{k=1}^K\alpha_k^2\sigma_k^2\nonumber\\
   & \leq-\frac{\eta^t\tau}{2}\left(1-4KS\eta^t \frac{\tau^2}{\tau}\sum_{k=1}^K\alpha_k^2\right)\left\Vert\nabla_{\theta}f\left(\theta^{t}\right)\right\Vert^2 +2KS\left(\eta^t\right)^2\tau^2\sum_{k=1}^K\alpha_k^2\left(\sigma_k^2+\epsilon^2\right)\nonumber\\
   & \leq-\frac{\eta^t\tau}{4}\left\Vert\nabla_{\theta}f\left(\theta^{t}\right)\right\Vert^2 +2KS\left(\eta^t\right)^2\tau^2\sum_{k=1}^K\alpha_k^2\left(\sigma_k^2+\epsilon^2\right),
\end{align}
where we first let $\eta^t \leq \frac{1}{16S\tau}$ and then let $\eta^t \leq \frac{1}{8SK\frac{\tau^2}{\tau}\sum_{k=1}^K\alpha_k^2}$. We also use $\left\Vert \nabla_{\theta}f\left(\theta^{t}\right)\right\Vert^2=\left\Vert \nabla_{\theta_c}f\left(\theta^{t}\right)\right\Vert^2+\left\Vert \nabla_{\theta_s}f\left(\theta^{t}\right)\right\Vert^2$.

Rearranging the above, we have
\begin{align}
    &\eta^t\left\Vert\nabla_{\theta}f\left(\theta^{t}\right)\right\Vert^2\leq \frac{4}{\tau}\left(f\left(\theta^{t}\right)-\mathbb{E} \left[f\left(\theta^{t+1}_s\right)\right]\right)+8KS\left(\eta^t\right)^2 \frac{2\tau^2}{\tau}  \sum_{k=1}^K\alpha_k^2\left(\sigma_k^2+\epsilon^2\right).
\end{align}
Taking expectation and averaging over all $t$, we have
\begin{align}
    &\frac{1}{T}\sum_{t=0}^{T-1}\eta^t\mathbb{E}\left[\left\Vert\nabla_{\theta}f\left(\theta^{t}\right)\right\Vert^2\right]\leq \frac{4}{T\tau}\left(f\left(\theta(0)\right)-f^\ast\right)+\frac{16KS\tau}{T} \sum_{k=1}^K\alpha_k^2\left(\sigma_k^2+\epsilon^2\right)\sum_{t=0}^{T-1}\left(\eta^t\right)^2.
\end{align}
Hence, we finish the proof of Proposition \ref{SFL-V1-invariability}.
\end{appendix}

\end{document}